%% file: main.tex
\documentclass{article}
\usepackage{graphicx} %
\usepackage{amsmath}
\usepackage{amssymb}
\usepackage{amsthm} %
\usepackage{amsthm}
\usepackage{graphicx}
\theoremstyle{definition}
\newtheorem{definition}{Definition}
\usepackage{color, braket}
\definecolor{nicepink}{rgb}{1.0, 0.33, 0.64}
\usepackage{booktabs}
\usepackage{array}
\usepackage{siunitx}
\usepackage[table]{xcolor}
\definecolor{lightgray}{gray}{0.85}
\usepackage{appendix}
\usepackage{subcaption}
\usepackage{graphicx}
\usepackage{tikz}

\usepackage{calc}
\usepackage{accents}
\newcommand{\doubletilde}[1]{\tilde{\raisebox{0pt}[0.85\height]{$\tilde{#1}$}}}

\usepackage{tikz}
\usepackage{xcolor}
\usepackage{quantikz}
\usepackage{hyperref}
\usepackage{xcolor}
\usepackage{quantikz}
\usepackage{pdflscape}   %
\usepackage{amssymb}
\newtheorem{theorem}{Theorem}

\newtheorem{note}[theorem]{Note}
\usepackage{algpseudocode}
\usepackage{algorithm}
\usepackage{authblk}

\usepackage[a4paper, margin=2.5cm]{geometry} %

\newcommand{\Span}{\text{Span}}
\newcommand{\mc}[1]{\mathcal{#1}}
\newcommand{\eff}{\text{eff}}
\newcommand{\prep}{\text{Prep}}
\newcommand{\meas}{\text{Meas}}

\graphicspath{{figures/}{./}} %

\definecolor{teal}{rgb}{0.0, 0.5, 0.5}
\definecolor{etonblue}{rgb}{0.59, 0.78, 0.64}
\usetikzlibrary{
  arrows.meta,      %
  positioning,      %
  calc,             %
  backgrounds,      %
  fit,              %
  shapes,           %
  shadows.blur      %
}

\definecolor{ink}{RGB}{45,45,45}
\definecolor{accent}{RGB}{0,102,204}
\definecolor{soft}{RGB}{230,236,245}
\definecolor{soft2}{RGB}{242,242,242}

\tikzset{
  every picture/.style={
    line cap=round, line join=round,
  },
  node/.style={
    font=\sffamily\small,
    align=center,
  },
  box/.style={
    node,
    draw=ink!60,
    fill=white,
    rounded corners=2pt,
    inner sep=6pt,
    minimum width=2.6cm,
    semithick,
    blur shadow, %
  },
  emph/.style={
    box,
    draw=accent!70!black,
    fill=accent!6,
  },
  link/.style={
    -{Latex[length=2.2mm]},
    semithick,
    draw=ink!80,
  },
  dashedlink/.style={
    link, dashed, draw=ink!50
  },
  label/.style={
    node, fill=soft2, draw=ink!15, rounded corners=1.5pt, inner sep=3pt
  }
}
\begin{document}

\title{Hybrid Method of Efficient Simulation of Physics Applications for a Quantum Computer}
\author[1,2]{Carla Rieger\thanks{These authors contributed equally to this work.}}
\author[3]{Albert T. Schmitz\protect\footnotemark[\value{footnote}]}
\author[4]{Gehad Salem} 
\author[3]{Massimiliano Incudini}
\author[1]{Sofia Vallecorsa} 
\author[3]{Anne Y. Matsuura} 
\author[1]{Michele Grossi}
\author[3]{Gian Giacomo Guerreschi} 

\affil[1]{European Organization for Nuclear Research (CERN), Geneva 1211, Switzerland}
\affil[2]{School of Engineering and Design, Technical University of Munich, Germany}
\affil[3]{Intel Labs, Intel Corporation, 2200 Mission College Blvd, Santa Clara, CA 95054, United States of America}
\affil[4]{The American University in Cairo, AUC Avenue, New Cairo 11835, Egypt}

\maketitle

\begin{abstract}
    Quantum chemistry and materials science are among the most promising areas for demonstrating algorithmic quantum advantage and quantum utility due to their inherent quantum mechanical nature. Still, large-scale simulations of quantum circuits are essential for determining the problem size at which quantum solutions outperform classical methods.
    In this work, we present a novel hybrid simulation approach, forming a hybrid of a fullstate and a Clifford simulator, specifically designed to address the computational challenges associated with the time evolution of quantum chemistry Hamiltonians. Our method focuses on the efficient emulation of multi-qubit rotations, a critical component of Trotterized Hamiltonian evolution. By optimizing the representation and execution of multi-qubit operations leveraging the Pauli frame, our approach significantly reduces the computational cost of simulating quantum circuits, enabling more efficient simulations.
    Beyond its impact on chemistry applications, our emulation strategy has broad implications for any computational workload that relies heavily on multi-qubit rotations. By increasing the efficiency of quantum simulations, our method facilitates more accurate and cost-effective studies of complex quantum systems. We quantify the performance improvements and computational savings for this emulation strategy, and we obtain a speedup of a factor $\approx 18$ ($\approx 22$ with MPI) for our evaluated chemistry Hamiltonians with 24 qubits. Thus, we evaluate our integration of this emulation strategy into the Intel Quantum SDK, further bridging the gap between theoretical algorithm development and practical quantum software implementations.
\end{abstract}

\section{Introduction}
\begin{figure}[t!]
    \centering
    \includegraphics[width=0.75\linewidth]{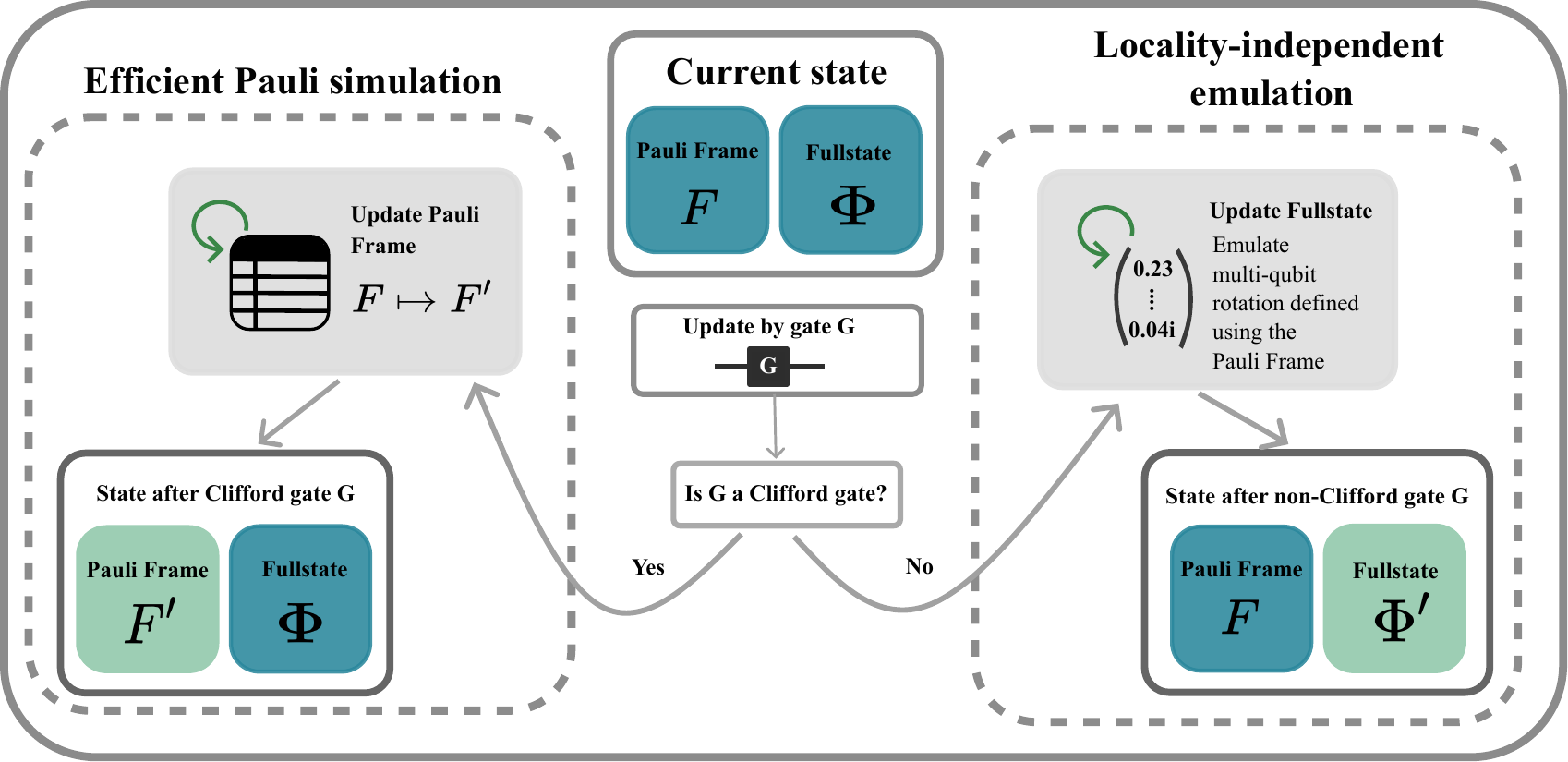}
    \caption{Workings of the hybrid approach combining Clifford and fullstate simulation. This includes, on the one hand, the efficient simulation of Pauli gates using Pauli frame tracking and the emulation of non-Clifford gates by utilizing the Pauli frame as a lookup table, which allows for representing multi-qubit Pauli rotations as single-qubit Pauli rotations with a modified rotation axis. The updated element is highlighted in green, respectively.}
    \label{fig:schematic}
\end{figure}
There is a growing consensus that applications in chemistry and materials science will be among the first to deliver quantum algorithmic advantage, in particular for quantum simulation tasks~\cite{georgescu2014quantum, motta2022emerging, kivlichan2020improved, campbell2021early, lee2023evaluating, o2016scalable}. Their formulation naturally leverages multi-qubit operations in the form of rotations $\exp(-i \frac{\theta}{2} P)$ generated by a tensor product of Pauli matrices, $P$, acting non-trivially on a subset of the qubits. This primitive is at the core of Trotter-Suzuki-based early fault-tolerant simulations, as well as the approximate quantum optimization algorithm~\cite{farhi2014quantum} and variational quantum eigensolvers~\cite{kandala2017hardware} for near-term devices.
It is crucial to emphasize the significance of classical simulation of complex quantum systems and dynamics in selecting promising quantum applications, verifying the performance of quantum computers, and identifying the break-even point at which the capabilities of a quantum device surpass those of classical computers. 

The quantum circuit framework expresses quantum programs in terms of single- and two-qubit gates, often obtained as decompositions of multi-qubit operations into logically equivalent but longer circuits made of simpler gates. This representation also reflects how most quantum computers are expected to work, with physical interactions often limited to two-body terms. While classical simulators process quantum circuits gate-by-gate, this is not strictly necessary to obtain the result of the computation. A family of programs known as quantum computer \textit{emulators}~\cite{haner2016high} leverages higher-level subroutines to reproduce the respective quantum logic and, by replacing gate-to-gate calculations with optimized shortcuts, may achieve substantial speed-ups. Intuitively, an emulator is only required to produce the same result as an ideal quantum computer without forcefully passing through all its intermediate configurations. Therefore, the requirements for a \textit{simulator} are stricter because every single gate-based operation is calculated \textit{exactly}, and noise effects or errors may be included in the calculation. In the context of this work, both the simulator and the emulator are executed on classical computing hardware.

In this work, we present a quantum computer emulator that can drastically speed up multi-qubit rotations, rendering them as \emph{resource-intensive as single-qubit rotations}. The improvement is most significant for Hamiltonians with high locality terms, and thus is relevant for Hamiltonians with a relatively high number of qubits. This speed-up requires two ingredients: a simulator backend with native multi-qubit operations, and a way to reconstruct them from circuits composed of 1- and 2-qubit gates.
Solving the latter aspect allows integrating this novel technique into existing quantum computing frameworks without the need to modify their current quantum intermediate representation (QIR). In fact, current QIRs are often centered around single- and two-qubit operations, making it difficult to deal directly with multi-qubit rotations without a major update to the framework itself.
We overcome this limitation by designing a new hybrid platform, which fuses the multi-qubit rotation fullstate simulator with a Clifford simulator. The hybrid simulator emulates all Clifford operations by updating the Pauli frame and applies non-Clifford operations as multi-qubit rotations, which are implemented in the fullstate simulator. In this approach, the Pauli frame handles the Clifford portion of the low-level decomposition of multi-qubit rotations without affecting the state of the quantum system and at a negligible computational cost.

We implemented our novel emulation approach in the Intel QSDK \cite{guerreschi2020intel, paykin2023pcoast} as the Clifford Fullstate Hybrid Simulator (CFHS) and benchmark its performance on Hamiltonian simulation workloads, with both artificially generated and quantum chemistry Hamiltonians. In this context, we generate random Hamiltonians with a defined locality on the one hand and select a set of real-world chemistry Hamiltonians ranging from 14 to 24 qubits on the other hand. Based on those Hamiltonian workloads, we investigate the scaling of runtime and compilation time of our novel emulation strategy and compare it to the previous version of the Intel Quantum Simulator (IQS) implementation. We observe a substantial decrease in the runtime of CFHS compared to the original IQS version, which becomes more pronounced with increasing locality and number of qubits, allowing for higher locality terms. In the exemplary case of chemistry Hamiltonians with 24 qubits and a mean locality of approximately $11$ we obtain a speedup of a factor of $\approx 18$ and $\approx 22$ when leveraging the message-passing interface (MPI). The runtime of CFHS does not depend on the locality of the terms in the Hamiltonian, while the linear scaling of the running time in the number of terms and exponential scaling in qubit number is preserved. By tracking the compilation time of CFHS and IQS, we confirm that the computational cost is not shifted from runtime to compilation time. 

This work is structured as follows. We start by summarizing the details on Clifford simulation in Section~\ref{sec:background}. Then we depict how multi-qubit rotations can be expressed in a fullstate simulator in Section~\ref{sec:fullstatesim} and conclude with the presentation of the Clifford fullstate hybrid simulator (CFHS) in Section~\ref{sec:hybridsim}, where the Pauli frame is used as a lookup table to efficiently emulate multi-qubit rotations. We benchmark our novel CFHS against the previous IQS version using random Hamiltonians and chemistry-specific molecular Hamiltonians, and present our findings in Section~\ref{sec:benchmark}. Then, we conclude this work with Section~\ref{sec:conclusion}.

\vspace{1cm}
\input{sec_background}

\input{sec_fullstate}

\input{sec_cfsh}

\input{sec_bench}

\section{Conclusion}\label{sec:conclusion}
\input{sec_conclusion}

\section*{Code availability}
The dataset and code supporting the findings of this work are available on \href{https://github.com/carlasophie/CFHS_simulation}{GitHub}.

\section*{Acknowledgments}
CR thanks Robert König for fruitful discussions regarding symplectic transvections, and Nicolas Sawaya and Daan Camps for their support with the HamLib dataset. CR is sponsored by the Wolfgang Gentner Programme of the German Federal Ministry of Education and Research (grant no. 13E18CHA). CR, SV, and MG are supported by the CERN through the CERN Quantum Technology Initiative.

\appendix
\input{sec_appendix}

\bibliographystyle{plain}
\bibliography{mybib}

\end{document}

%% file: sec_background.tex
\section{Elements of  Clifford simulation and notation} \label{sec:background}

The development of simulation techniques has a rich history within the study of quantum computation and the classical simulability of quantum systems. In this section, we focus on Clifford simulation, and in the next one, we describe how to implement multi-qubit rotations in fullstate simulators. The two concepts will be combined in Section~\ref{sec:hybridsim}.

Clifford simulation originates from the structure of the Clifford group, and plays a crucial role in quantum error correction~\cite{kitaev1997quantum}, quantum teleportation~\cite{bennett1993teleporting}, and, most prominently here, the efficient classical simulation of certain quantum circuits~\cite{gottesman1997stabilizer}. In detail, the Gottesman-Knill theorem states that a quantum circuit consisting only of Clifford gates from Eq.~\eqref{eqn:clifford_group} and qubits initially prepared and measured in the computational basis is classically simulable in polynomial time~\cite{gottesman1998heisenberg}. The ability to efficiently simulate Clifford circuits has had profound implications for quantum error correction (e.g., the surface code~\cite{kitaev1997quantum, kitaev2003fault}), and benchmarking quantum hardware~\cite{magesan2011scalable, knill2008randomized}. Interestingly, while Clifford circuits can, for example, generate highly entangled states, e.g., maximally entangled Bell states, they do not provide exponential (computational) speedup over classical computation. The Gottesman-Knill theorem highlights the limitations of quantum computational advantage within the Clifford framework while enabling practical simulation methods for an essential class of quantum circuits that can be extended to universality by adding a T-gate.
In the following, we will first define the Pauli and Clifford groups, then proceed with the definition of Pauli masks, and conclude with more details on simulating Clifford circuits. 

\subsection{Pauli and Clifford Group}

Aiming at defining the Clifford group, we recall the definition of Pauli matrices:
\begin{equation}
    \sigma_0 = \mathbb{I} , \, \, \sigma_1 = X = \begin{pmatrix}
        0 & 1 \\
        1 & 0
    \end{pmatrix}\, , \, \,
    \sigma_2 = Y = \begin{pmatrix}
        0 & -i \\
        i & 0
    \end{pmatrix}\, , \, \,
    \sigma_3 = Z = \begin{pmatrix}
        1 & 0 \\
        0 & -1
    \end{pmatrix}\, . \, \,
    \label{eqn:pauli_mat}
\end{equation}
Note that $XYZ = i \mathbb{I}$ and that for all $j=0,1,2,3$ we have $\sigma_j^2 = \mathbb{I}$. As a consequence, $\sigma_i$ and $\sigma_j$ either commute or anti-commute. We define the anti-commutation flag function as follows:
\begin{equation}
    \lambda(\sigma_i, \sigma_j) = \begin{cases}
        0, & \text{if } \sigma_i, \sigma_j \text{ commute} \\
        1, & \text{otherwise} \, .
    \end{cases} 
\end{equation}
Based on the Pauli matrices in Eq.~\eqref{eqn:pauli_mat}, we define the Pauli group $\mathcal{P}_n$ on $n$ qubits (see e.g.~\cite{nielsen2010quantum}) and in  the one-qubit case $n=1$, we have:
\begin{equation}
    \mathcal{P}_1 =\{ \pm \mathbb{I}, \pm i \mathbb{I}, \pm X, \pm i X,  \pm Y, \pm i Y, \pm Z, \pm i Z \} \equiv \langle X, Y, Z \rangle \, .
    \label{eqn:pauli_group}
\end{equation}
With the respective multiplicative factors $\pm 1$ and $\pm i$, the group~$\mathcal{P}_1$ defined in Eq.~\eqref{eqn:pauli_group} is closed under multiplication. The group $\mathcal{P}_1$ can be generalized for $n$ qubits by taking the $n$-folded tensor product of matrices~$\{ \mathbb{I}, X, Y, Z \}$ %
and including the multiplicative factors $\pm 1$ and $\pm i$. The anticommutation flag function is extended to the $n$-qubit case, 
\begin{equation}
    \lambda(P_1, P_2) = \begin{cases}
        0, & \text{if } P_1, P_2 \text{ commute} \\
        1, & \text{otherwise}
    \end{cases} \, \, \,  = \left(  \sum_{j=0}^{n-1} \lambda(P^1_j, P^2_j) \right) \text{ mod } 2
\end{equation}
Then, the Clifford group $\mathcal{C}_n$ is defined as the group of unitaries normalizing the Pauli group $\mathcal{P}_n$~\cite{gottesman1998heisenberg}: 
\begin{equation}
    \mathcal{C}_n = \{A \in U_{2^n} | \, A \, P \, A^{\dagger} \, \in \, \mathcal{P}_n  \,\, \forall \,  P \in \mathcal{P}_n \} \, \backslash U(1) \, . \label{eqn:clifford_group}
\end{equation}
The Clifford group is generated by the gate set $\langle \, H, S, CNOT \,\rangle$, with the gates being defined as 
\begin{equation}
    H = \frac{1}{\sqrt{2}} \begin{pmatrix}
        1 & 1 \\
        1 & -1
    \end{pmatrix}, \,  \,
    S = \begin{pmatrix}
        1 & 0 \\
        0 & i
    \end{pmatrix}, \,  \,
    CNOT =  \begin{pmatrix}
        1 & 0 \\
        0 & 0
    \end{pmatrix} \otimes \mathbb{I} \, + \,
    \begin{pmatrix}
        0 & 0 \\
        0 & 1
    \end{pmatrix} \otimes X \, .
\end{equation}
\subsection{Pauli Masks}

We identify $n$-bit integers between $0$ and $2^{n}-1$ with their $n$-bit string following the computer science notation, with the leftmost bit being the most significant one. Then, the $n$-bit $X$-mask of an $n$-qubit Pauli string $P = \otimes_{j=0}^{n-1} \, \sigma_{i_j}$ is defined as
\begin{equation}
    m_X(P) := b_{n-1} \ldots b_0 \, , \qquad \text{with  } b_j = \begin{cases}
        1 \, , & \text{if } \sigma_{i_j} = X (=\sigma_1) \\
        0 \, , & \text{otherwise} \, . 
    \end{cases} 
    \label{def:xmask}
\end{equation}
The $Y$-mask and $Z$-mask of $P$, respectively denoted by $m_Y (P)$ and $m_Z (P)$, are defined analogously. We denote the addition modulo~2 (i.e., a XOR) between a pair of bits by $\oplus$. Then, the bit-wise XOR between a pair of $n$-bit strings is denoted by $\odot$ to differentiate it from a XOR: 
\begin{align}
    \cdot \odot \cdot &: \{0,1\}^n \times \{0,1\}^n \to \mathbb{N} \\  
    a \odot b & = (a_{n-1} \ldots a_0) \odot (b_{n-1} \ldots b_0) = \sum_{j=0}^{n-1} (a_j \oplus b_j) \, 2^j.
    \label{def:bitwise_xor}
\end{align}
As usual, the Hamming weight of a $n$-bit binary string, denoted by~$|\,\cdot\,| \, $, is defined as:
\begin{equation}
    |b| = |b_{n-1} \ldots b_0| = \sum_{j=0}^{n-1} b_j \, . 
    \label{def:hamming_weight}
\end{equation}
\subsection{Details on Clifford Simulation}\label{sec:clifsim}

Going more into the details, Clifford simulation is highly efficient because every stabilizer state acting on $n$ qubits only requires storing $n$ multiplicatively-independent Pauli operators, each of which is at worst $\mc O(n)$ to store. These operators generate the so-called stabilizer group, which uniquely defines the state as the simultaneous $(+1)$-eigenstate of all elements of the commutative group. As the typical fiducial initial state $\ket{0}\equiv\ket{0}^{\otimes n}$ is the stabilizer state of the stabilizer group $\Span \left(\{Z_j\}_{j<n}\right)$, and Clifford operators and gates, by definition, transforms Pauli operators to other Pauli operators, the action of a Clifford on the stabilizer state translates to the transformation of the initial Z-basis stabilizer set under conjugation. For localized Clifford gates, a Pauli operator can be updated under the action of the Clifford gate in $\mc O(1)$, and thus each gate update is at worst $\mc O(n)$. Keeping track of the stabilizer group is sufficient for purely unitary operations, but one also needs to extract and simulate Pauli measurements classically. It is easy to show that the expectation of a Pauli operator $P$ with respect to a stabilizer state is entirely dependent on whether $P$ is in the stabilizer group up to a phase, or not. If $P$ modulus of the phase is in the group, the expectation value is the residual phase. This translates to a measurement outcome of $\pm 1$ for a measurement of $P$ based on its phase relative to the group. If $P$ is not in the group modulus of the phase, then the expectation value is $0$, resulting in a mixed measurement outcome. So to evaluate the measurement outcome, one needs the ability to expand $P$ as a product of elements of the stabilizer set. Though other methods exist, the most efficient one is to also track the conjugation of a set of \emph{destabilizers}, conventionally starting with $\{X_j\}_{j<n}$. The resulting set of $2n$ operators, is referred to here as the \emph{Pauli frame}\footnote{The more typical term for this object is the \emph{Pauli Tableau}, but we use ``frame'' here to emphasize the frame property of Eq.~\eqref{eq:expand}.}, which for our purposes we will arrange as a $n \times 2$ matrix:
\begin{align}\label{eq:framemat} 
F = \begin{pmatrix}
\eff Z_0 & \eff X_0 \\
\vdots & \vdots \\
\eff Z_{n-1} & \eff X_{n-1}
\end{pmatrix} \, ,
\end{align} 
\noindent
where $\eff Z_j, \eff X_j$ is the result of conjugating the initial stabilizers and destabilizers by the Clifford unitary. We also use $\eff Y_j = -i \, \eff Z_j * \eff X_j$ and $F_{origin}$ to represent the \emph{origin} frame, i.e., the frame containing all single-qubit $Z$ and $X$ operators in the row corresponding to their qubit index. The Pauli frame has the property that any Pauli operator can be expanded as a product of its elements via the formula, 
\begin{align}
P \propto \prod_{j<n} (-i)^{\lambda(P, \, \eff X_j) \oplus \lambda(P, \, \eff Z_j)} \, \eff Z_j^{\lambda(P, \,  \eff X_j)} \,  \eff X_j^{\lambda(P, \, \eff Z_j)} \, . \label{eq:expand}
\end{align}
So if $\{\eff Z_j\}_{j<n}$ represents the set of stabilizers of our stabilizer state, the expansion of $P$ mod phase is a matter of computing $\lambda(P, \eff X_j)$ and $\lambda(P, \eff Z_j)$, each of which requires at worst $\mc O(n)$ to compute, implying a $\mc O(n^2)$ cost to simulate a Pauli measurement.\footnote{We note that a measurement requires more than just an assessment of the measurement outcome but also state collapse in the case that the measurement does not commute with the stabilizer group. We avoid discussing how this state collapse is implemented for Clifford simulation, as it is unnecessary for our purposes, where state collapse is simulated on the fullstate side of the hybrid simulation discussed in the text. } 

\begin{figure}[h!]
    \centering
    \includegraphics[width=0.5\linewidth]{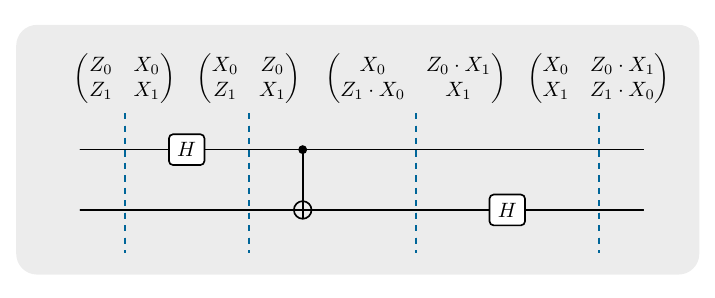}
    \caption{Demonstration of the sequential updates applied to the respective Pauli frame as defined in Eq.~\eqref{eq:framemat} under the gate-by-gate action of the Cliffords. Thus, by tracking updates to the Pauli frame, we can efficiently simulate this Clifford circuit.}
    \label{fig:cliffordupdates_tableau}
\end{figure}

%% file: sec_fullstate.tex
\section{Multi-qubit rotations in a fullstate simulator} \label{sec:fullstatesim}

The intuitive idea for implementing efficient multi-qubit rotations in a fullstate simulator stems from the observation that computational basis states are updated in pairs, regardless of the multi-qubit Pauli operator that generates the rotation. The only difference between one-qubit and multi-qubit rotations is in the pattern of such pairs. Consider a $n$-qubit system and the multi-qubit Pauli operator $P$ (in the following, we will refer to it simply as a \textit{Pauli}). Including the identity, $P$ can be seen as a tensor product of $n$ Pauli matrices:
\begin{equation}
    P = \bigotimes_{j=0}^{n-1} \sigma_{i_j} \, ,
\end{equation}
with $i_j \in \{0,1,2,3\}$. %
Now, consider an arbitrary computational basis state:
\begin{equation}
    \ket{k} = \bigotimes_{j=0}^{n-1} \ket{k_j} \, ,
\end{equation}
with $\{k_j\}_j$ being the bits in the binary representation of $k \in \{0, 1, 2, \dots, 2^n-1\}$ such that the lowest significance bit is associated with the state of qubit 0, as usually assumed in the Computer Science~(CS) community in contrast to the Quantum Information community that often adopts the opposite convention. Here, we follow the CS convention since, among other things, it is at the basis of the Intel Quantum Simulator, \emph{i.e.} of the software we used to implement the multi-qubit rotations.\\
\noindent
Denoting the bit-wise XOR of two $n$-bit string as in Eq.~\eqref{def:bitwise_xor}
and the Hamming weight of $n$-bit string as in Eq.~\eqref{def:hamming_weight},
the action of $P$ on $\ket{k}$ can be expressed as:
\begin{align}
    P \ket{k} =& \bigotimes_{j=0}^{n-1} \sigma_{i_j} \ket{k_j} \\
    =& \,(i)^{|m_Y(P)|} (-1)^{|m_Y(P) \, \odot \, k| + |m_Z(P) \, \odot \, k|} \ket{k \odot m_X(P) \odot m_Y(P)} \\
    \equiv&\,  e^{i \varphi_P(k)} \ket{f_P(k)} \, ,
\end{align}
where we used the definition of  
\begin{align}
    f_P(k) := k \, \odot \, m_X(P) \, \odot \, m_Y(P) \, ,\nonumber
\end{align} and
\begin{align}
\varphi_P(k) := \tfrac{\pi}{2} |m_Y(P)| + \pi (|m_Y(P) \, \odot \, k| + |m_Z(P)) \, \odot \, k|) \, . \nonumber
\end{align}
Therefore, the action of the multi-qubit rotation $R_P(\theta)=e^{-i \theta P/2}$ on a generic computational basis state~$\ket{k}$ can be expressed as:
\begin{align}
    R_P(\theta) \ket{k} =& \exp{(-i \theta P/2)} \ket{k} \\
    =& \cos{(\theta/2)} \ket{k} -i \sin{(\theta/2)} P \ket{k} \nonumber \\
    =& \cos{(\theta/2)} \ket{k} -i \sin{(\theta/2)}
    e^{i \varphi_P(k)} \ket{f_P(k)} \, .
\end{align}
The initial observation that computational basis states can be updated in pairs comes from the fact that 
\begin{align}
f_P(f_P(k)) = k \, .
\end{align}

We describe the update rule for the amplitude of generic qubit states (represented as an array of complex numbers in the fullstate simulator) in Appendix~\ref{sec:appendix_update_rule_multiqubit_rotation}.
The implementation in a fullstate simulator is relatively straightforward.
The situation becomes slightly more complicated when distributed storage and computation are utilized to accelerate simulations. As a concrete scenario, we consider the Intel Quantum Simulator~\cite{guerreschi2020intel,smelyanskiy2016qhipster}, a high-performance computing software designed for simulating quantum computers.
In this case, the complex vector storing the state amplitudes is distributed across the local memory of multiple MPI processes, and one needs to communicate the relevant part of the state before updating the amplitudes.
By dividing the state into equally large portions across a number of processes that is a power of 2, not only are the amplitudes updated in pairs, but the ``partner'' of every amplitude in process $j$ either belongs to the same process $j$ or to the partner process $j^\prime$ (i.e. communication between pairs of processes\footnote{In the language of distributed fullstate simulators, $j^\prime=f(j, \tilde P)$ with $\tilde P$ being the restriction of $P$ to the global qubits only.} is sufficient).
This simplifies the MPI communication pattern and maintains an efficient implementation.

%% file: sec_cfsh.tex
\section{Clifford-Fullstate Hybrid Simulator} \label{sec:hybridsim}

\subsection{Pauli Frame as a Lookup Table for Multi-qubit Rotations and Measurements}

The Pauli frame is a more interesting object beyond Clifford simulation. In particular, it is a unique representation of the Clifford unitary, up to a global phase. That is, any Clifford unitary mod phase can be represented by a unique Pauli frame and any collection of $2n$ Pauli operators arranged as in Eq.~\eqref{eq:framemat} and which satisfy $\lambda(\eff Z_i, \eff Z_j) = \lambda(\eff X_i, \eff X_j) = 0$, $\lambda(\eff Z_i, \eff X_j) = \delta_{ij}$ represents a unique Clifford unitary. However, this requires one to fix the interpretation. The conventional interpretation for Clifford simulation, which we deem the \emph{forward interpretation}, asserts that a frame $F$ maps to a Clifford unitary $U$, via the relation $F= UF_{origin} U^\dagger$, where the conjugation is entry-wise. However, there is another interpretation, the \emph{backward interpretation}, which asserts that $F = U^\dagger F_{origin} U$. Both interpretations can be applied to Clifford simulation, and although similar in form, the major difference between using the backward interpretation and the forward interpretation lies in how the frame is updated by a gate. Instead of simply conjugating each element by the gate, each gate has an update rule whereby the new entries of the frame are products of the previous frame. Where for forward interpretation, each gate was a $\mc O(1)$ update to a $\mc O(n)$ entries, in backward interpretation, most gates are an $\mc O(n)$ update (to multiply at most two arbitrary Pauli operators) to $\mc O(1)$ number of entries. Thus, the cost for gate updates is the same in both interpretations, but there are some advantages for backward updates in the case of sparse Pauli operator representations, as it preserves the sparsity.

The advantage of the backward interpretation is that it makes translating single-qubit rotations and measurements to their multi-qubit equivalents in the context of Clifford gates trivial. To demonstrate this, suppose we have a sequence of Clifford gates which collectively represent the unitary $U$ followed by a rotation gate $R_Z(i,\theta)$ applied to qubit $i$. We can effectively commute $U$ past the rotation at the cost of transforming the rotation axis to that of a multi-qubit rotation:
\begin{align}
R_Z(i, \theta) U = U U^\dagger R_Z(i, \theta) U = U R_{\eff Z_i} (\theta). \label{eq:single_to_multirot}
\end{align}
where $\eff Z_i$ is the $(i^{th}, 0)$ entry of a Pauli frame representing $U$ in the backward interpretation. So the Pauli frame in the backward interpretation represents a ``lookup'' table for how single-qubit operations translate in the context of the Clifford unitary that came before it. This is used to translate circuits to an abstract Pauli-based graphical representation for the same purpose of optimization and analysis in PCOAST~\cite{paykin2023pcoast, schmitz2023optimization}, but in our context, we marry this idea with those of Section~\ref{sec:fullstatesim} for the purposes of simulation.

\begin{figure}[h!]
    \centering
    \includegraphics[width=0.5\linewidth]{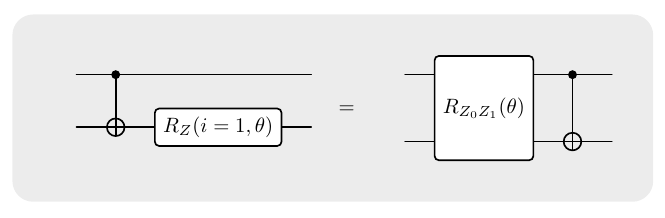}
    \caption{Illustrating the commutation behavior as described in Eq.~\eqref{eq:single_to_multirot}, we can commute the Clifford unitary $U$ (here given by a single CNOT gate as an example) past the rotation, translating the single qubit rotation to a multi-qubit one. }
    \label{fig:single_to_multirot}
\end{figure}

\subsection{Description of the Clifford-fullstate hybrid simulation}

We now have all the tools to describe the Clifford-fullstate hybrid simulation (CFHS). We write the CFHS as $(F, \Phi)$, where $F$ is a Pauli frame representing a Clifford unitary $U$ in the \textit{backward} interpretation, and $\Phi =\{\phi_k\}$ is a fullstate simulator representing the state $\ket{\phi}$ in the computational basis. The actual quantum state which $(F, \Phi)$ represents is 
\begin{align}
(F, \Phi) \simeq \ket{\psi} = U\ket{\phi} \, .   \label{eq:cfhsstate}
\end{align} 

The initial state of the simulator is $(F, \Phi) \gets (F_{origin}, \delta_{k,0})$ representing the fiducial 0 state, $\ket{\psi} = \ket{0}$. When the simulator is updated by a gate $g$, we follow the flow diagram of Figure~\ref{fig:gateupdateflow}. Most importantly, if~$g$ is a Clifford gate, it is used to update $F$, where any non-Clifford gate is used to update the fullstate simulator with the appropriate multi-qubit operation as provided by the Pauli frame. On the time scales of an update to the fullstate simulator, updates to the frame are effectively free, and as updates to the fullstate simulator for multi-qubit operations are flat with respect to the support of the Pauli operator, the simulation time for CFHS scales as the number of non-Clifford gates. Depending on the algorithm in question, this results in an $\mc O(1)$ to $\mc O(n)$ speed-up for an only $\mc O(n^2)$ \textit{additive} increase in memory, which is vanishingly small compared to the $\mc O(2^n)$ memory footprint of the fullstate simulator. Thus, this forms an exponential decrease to the memory footprint that would be needed by the fullstate simulator.

\subsection{Alternative Interpretation of the Clifford-Fullstate Hybrid Simulation}

It should be clear that there is no loss of generality for CFH simulation over the fullstate simulation. Moreover, we can understand how CFHS relaxes the constraints of Clifford simulation by considering a computational basis expansion of Eq.~\eqref{eq:cfhsstate}:
\begin{align}
(F, \Phi) \simeq  U\ket{\phi} = \sum_{k} \phi_k (U\ket{k}) \, .   
\end{align}

The set $\{U\ket{k}\}$ represents a basis of stabilizer states defined by stabilizer groups connected to each other by flipping the signs of the stabilizer basis elements -- as encoded in the Pauli frame -- corresponding to the 1-bit positions of $k$. For Clifford simulation -- and equivalently CFH simulation of only Clifford gates -- the simulator represents the state $U\ket{0}$, i.e., we can only represent a single state of that basis. CFH simulation is the natural extension of Clifford simulation to allow an arbitrary superposition over a stabilizer state basis, whose coefficients are stored in the fullstate vector $\Phi$. By this same analysis, the CFH simulation is also an extension of fullstate simulation. Fullstate simulation fixes the basis and only updates the coefficients, whereas CFHS is a method for judiciously choosing whether to update the coefficients or update the basis. In this case, we have a highly efficient way of storing and updating the basis for a limited set of unitary transformations, and an efficient method for determining the action of any other gate relative to this dynamic basis. One could imagine other cases outside of gate-based quantum computing where this hybrid simulation method could be useful, analogous to the interaction picture of quantum mechanics.

\subsection{Extraction of Expectation Values}
As discussed in Section~\ref{sec:fullstatesim}, we can extract Pauli operator expectation values, and the same holds for CFH simulation. To extract the expectation value of the Pauli operator $P\in \mathcal{P}_n$, one only has to map $P$ under the action of the Pauli frame represented by $U$, defined by
\begin{align}
F(P) = U^\dagger P U = \text{phase}(P)\prod_{i<n} (-i)^{\lambda(P, X_i) + \lambda(P, Z_i)}\eff Z_i^{\lambda(P, X_i)} \eff X_i^{\lambda(P, Z_i)} \, . \label{eq:Fmap}
\end{align}
The expectation value of the Pauli~$P$ is then given by
\begin{align}
\braket{P}_{(F, \Phi)} = \bra{\psi} P \ket{\psi} = \bra{\phi} F(P)\ket{\phi} = \braket{F(P)}_\Phi \, . \label{eq:chfsexp}
\end{align} 
This is sufficient to estimate the cost functions of variational algorithms, for example, VQE and QAOA.
When one is interested in extracting computational basis coefficients and probabilities, then the computation is less straightforward. The use of Eq.~\eqref{eq:chfsexp} to these ends is prohibitively expensive as it would require a sum over $2^n$ expectation values. Instead, the best method would be to apply the Clifford unitary to the fullstate, i.e., 
\begin{align}
(F, \Phi) \to (F_{origin}, \Psi) \, , \nonumber
\end{align}
and then extract the coefficients from $\Psi$. A Clifford unitary $U \in \mathcal{C}_n$ can be implemented using at most~$\mc O(n)$ multi-qubit rotations of the form $ e^{-i\pi/4 P}$ with Paulis $P \in \mathcal{P}_n$, which is formally proven in Appendix~\ref{sec:appendix_constralg}. Thus, any Pauli frame can be implemented/inverted using at most $\mc O(n)$ multi-qubit rotations. A constructive algorithm for doing this is provided in the appendix~\ref{sec:appendix_constralg}\footnote{ Note there will always be an ambiguity in the overall phase of the state. One must keep this in mind when considering the raw basis-state coefficients.}.

%% file: sec_bench.tex
\section{Applications and Benchmarking}
\label{sec:benchmark}
In this section, we present the analysis of the performance of our novel emulation approach, implemented as a backend in the Intel QSDK. We evaluate our CFH-simulator against the Intel Quantum Simulator (IQS) version~1.1.1, and additionally compare the results obtained with and without enabling parallelization through message-passing interface (MPI) protocols. By investigating a variety of Hamiltonian simulation workloads, including those stemming from quantum chemistry and random Hamiltonians with varying locality for different numbers of qubits, we ensure a comprehensive evaluation across a coherent and broad parameter range. The findings are summarized below, and further evaluation results can be found in Appendix~\ref{sec:appendix_eval}.

\subsection{Preliminaries of Hamiltonian Simulation}
For a given time-independent Hamiltonian $H$, based on Schrödinger's equation, the unitary $U(t) = e^{-iHt}$ defines the temporal evolution in time $t$ (note that we set $\hbar =1$). The Hamiltonian $H$ can be written as a sum of local Pauli terms $H = \sum_j h_j$, as they form an orthogonal basis. Each $h_j$ is proportional to a product of Paulis $\sigma_{j_k}$ with $j_k \in \{0,1,2,3\}$ up to a real factor $c_j \in \mathbb{R}$:
\begin{equation}
    h_j = c_j \, \bigotimes_k \, \{ \sigma_{j_k} \} \, .
    \label{eqn:klocal}
\end{equation}
In general, the individual local terms in the sum do not commute, i.e., $[h_i, h_j] \neq 0$ for $i\neq j$. We can break down the temporal evolution described by $U(t)$ in smaller temporal steps via the first-order Trotter-Suzuki~\cite{trotter1959product, suzuki1976generalized} decomposition, given by 
\begin{align}
   U(t) = e^{-iHt} = \text{exp} \left( -it\sum_j h_j \right) \approx  \left( \prod_j e^{-i h_j \frac{t}{M}} \right)^{M} \, , \label{eqn:trotter}
\end{align}
with $M \in \mathbb{Z}_{>0}$ denoting the number of Trotter steps. Each of the Trotter steps in the Trotter-Suzuki decomposition is typically implemented via a CNOT staircase~\cite{trotter1959product, whitfield2011simulation}, as depicted in the exemplary quantum circuit for implementing the unitary $e^{- i \frac{\theta}{2} \left(\sigma_x \otimes \sigma_y \otimes \sigma_z \otimes \sigma_x \right) }$ in Figure~\ref{fig:cnot_ladder}. Note that this implementation requires, in general, $2(k-1)$ CNOT gates where $k$ corresponds to the \textit{locality} of the Hamiltonian, which is formally introduced in the following:
\begin{figure}
    \centering
    \includegraphics[width=0.9\linewidth]{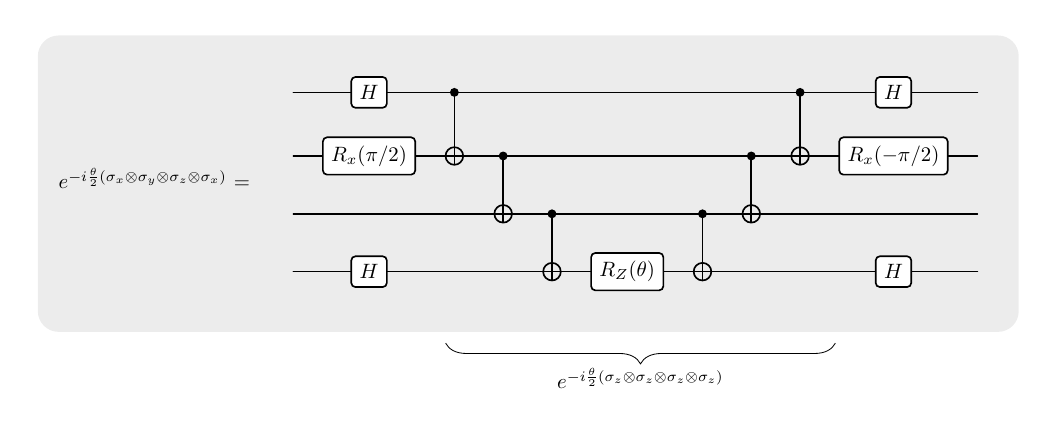}
    \caption{The unitary $e^{-i \frac{\theta}{2} P}$ generated by the Pauli $P=\sigma_x \otimes \sigma_y  \otimes \sigma_z  \otimes \sigma_x$ acting on 4 qubits can be decomposed in the presented CNOT staircase~\cite{trotter1959product, whitfield2011simulation} for implementation with one and two qubit gates.}\label{fig:cnot_ladder}
\end{figure}
\begin{definition}[\cite{kitaev2002classical, bluhm2024hamiltonian}] 
\label{def:locality}
A Hamiltonian $H \in \mathcal{H}^{\otimes n}$ acting on $n$ qubits has a \textit{locality} of $k$ or is named $k$-local, if $ H = \sum_i c_i P_i $ with $P_i \in \{ \mathbb{I}, X, Y, Z\}^{\otimes n}$ has $c_i= 0$ for all $|P_i| > k$, where,~$|P_i|$ denotes the Pauli weight of the respective Pauli string $P_i$. The Pauli weight is given by the count of non-trivially acting Paulis in the tensor product of $P_i$, respectively.
\end{definition}
Note that the respective decision problem of whether a Hamiltonian is of locality~$k$ is formalized as Hamiltonian Property Testing~\cite{bluhm2024hamiltonian}. The Pauli weight is a good representative of the resource cost for implementing the Pauli term in the era of both NISQ~\cite{preskill2018quantum} and fault-tolerant quantum algorithms~\cite{chiew2024defining, steudtner2018fermion}. 
Thus, it can be beneficial to reduce the Pauli weight of the Hamiltonian, and thereby, the choice of \textit{mapping} from first-quantization to a qubit-based format is relevant~\cite{chiew2024defining, sawaya2020resource}. Here, we focus on the Jordan-Wigner~\cite{jordan1928paulische} representation as given in~\cite{sawaya2024hamlib} for consistency, but we note that an alternative mapping could decrease the Pauli weight and thus further reduce resource costs.

\subsection{Evaluation on Random Hamiltonians}
To unveil the performance improvement in the simulation time of CFHS, we create a set of random Hamiltonians with distinct parameter configurations in our benchmarking setting. Hence, we evaluate CFHS and IQS on random Hamiltonians with either 50 or 100 terms and a number of qubits from 8 to 24. We increase the locality up to the maximal number of qubits in steps of~2 up to the respective qubit number of each Hamiltonian. The parameters are summarized in the Table~\ref{tab:params_rh}, and the specific set of parameters is chosen in order to have a feasible runtime for evaluation.
\noindent
To construct the random Hamiltonians used in our study, we create a number of $n_{\mathrm{terms}}$ unique non-trivially acting $k$-local Paulis and draw the real coefficients $c_j$ as in Eq.~\eqref{eqn:klocal} from the uniform distribution, which are then normalized. Thus, we can choose each unique term randomly from the set with the following number of elements:
\begin{align}
    \binom{n_{\mathrm{qubits}}}{k} \cdot 3^k \, .
\end{align}
As the randomly sampled Paulis are hermitian and the coefficients $c_j$ are real values, we get a hermitian Hamiltonian as in Eq.~\eqref{eqn:klocal} where we sum over the given number of terms $n_{\mathrm{terms}}$, corresponding to the given desired parameter configuration. \\
\begin{table}[h!]
\centering
\begin{tabular}{@{} l c c c r @{}}
\rowcolor{lightgray}
\textbf{Parameter} & \textbf{Acronym} & \textbf{Min. value} & \textbf{Max. value}  \\
No. qubits & $n_{\mathrm{qubits}}$  & 8   & 24 \\
\rowcolor{lightgray!40}
Locality  & $L$  & 4   & $n_{\mathrm{qubits}}$ \\
No. Terms  & $n_{\mathrm{terms}}$ & 50   & 100 \\
\bottomrule
\end{tabular}
\caption{Range of parameters used to generate the specific random Hamiltonians based on which we select a number of 280 distinct parameter configurations for generating each random Hamiltonian.}\label{tab:params_rh}
\end{table}

\subsubsection*{Compilation time} %
In addition to the simulation time, we also track the compilation time to demonstrate that the improvement in simulation time is not due to offloading the workload towards compilation. Our evaluation shows that the ratio of compilation time CFHS compared to IQS without MPI is on average $0.95$ with a standard deviation of $0.04$, whereas the ratio of compilation time CFHS compared to IQS with MPI is on average $0.98$ with a standard deviation of $0.03$. The detailed analysis can be found in Appendix~\ref{sec:appendix_comptime_eval}.

\subsubsection*{Simulation time} %

\begin{figure*}[h!]
    \centering
    \includegraphics[width=0.9\textwidth]{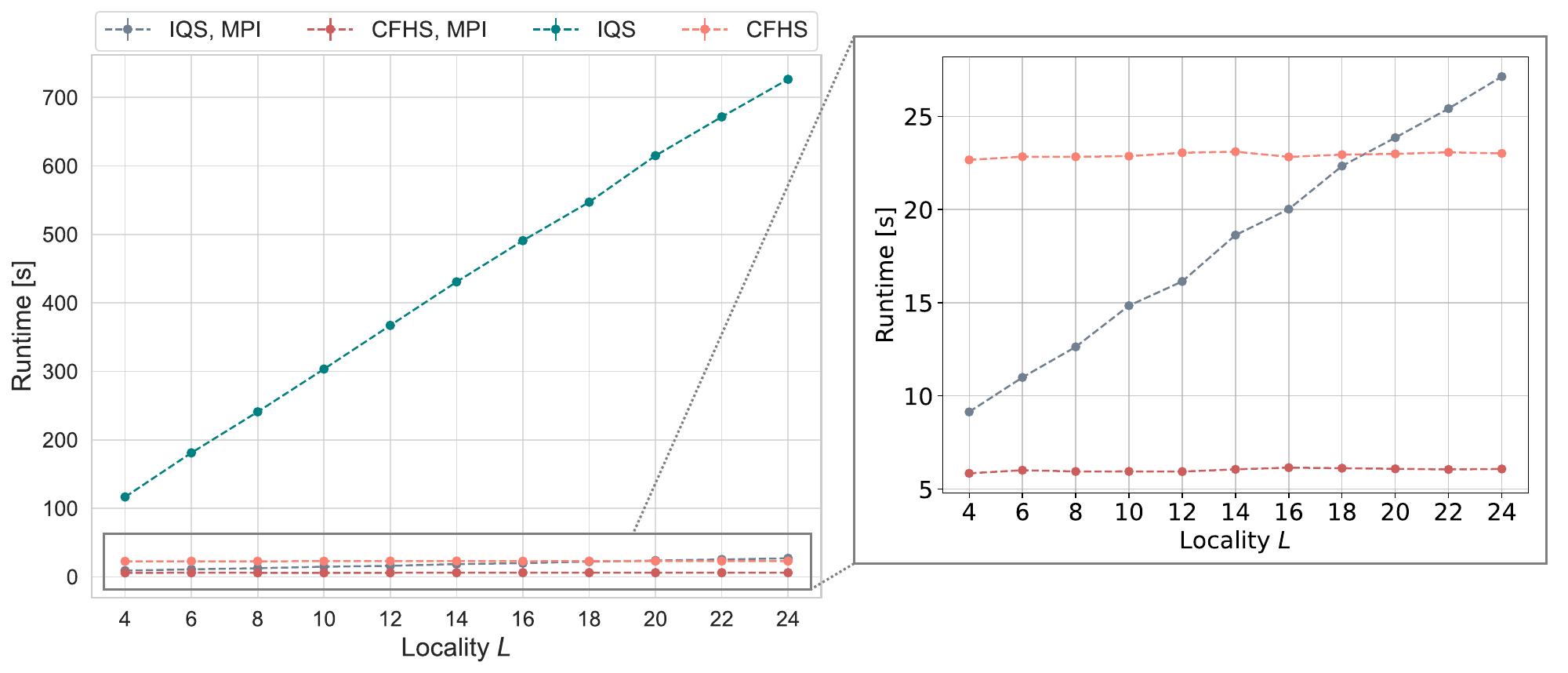}
    \caption{Comparison of the simulation time (runtime [s]) for random Hamiltonians with varying locality~$L$, 100 terms and 24 qubits, comparing results based on CFHS and IQS with and without MPI in this exemplary case. We observe a significant reduction in simulation time with our new CFHS method in both cases, with and without MPI. The linear increase of IQS simulation time in terms of the locality $L$, in contrast to the constant scaling of CFHS in locality~$L$, is nicely visible for results obtained with and without MPI. }\label{fig:sim_comparison24q_100}
\end{figure*}
We evaluate the simulation time (runtime) for the random Hamiltonian with a parameter range described in Table~\ref{tab:params_rh}. In order to visualize the improvement in scaling of CFHS over IQS, we showcase the runtime in dependence on the locality of the random Hamiltonians with 24 qubits and a locality ranging from 4 to 24. In Figure~\ref{fig:sim_comparison24q_100}, it is well visible that the simulation time of IQS increases linearly in the locality $L$, while the simulation time based on CFHS is constant in the locality $L$. In the exemplary case of a 24-qubit random Hamiltonian with 100 terms and a locality of 24, we obtain an improvement of runtime based on CFHS of $31.6$ compared to IQS without using MPI. 
The same scaling behavior of CFHS and IQS with locality $L$ is visible when enabling MPI, where we observe an improvement by a factor of $3.81 \pm 0.05$ for CFHS with MPI and $22.1 \pm 4.2$ for IQS with MPI in the case of 24 qubit Hamiltonians compared to the runtime without utilizing MPI. The overall speedup in runtime, in dependence on the number of qubits per random Hamiltonian, is depicted in Figure~\ref{fig:runtime_case24q_ratio_MPI}, where we take the mean and standard deviation over Hamiltonians with equal qubit number but different localities. Evaluating the ratio of runtime without MPI and the runtime with MPI, we see a greater improvement when utilizing MPI for the original IQS version, exhibiting a gap towards CFHS. This gap seems to become larger with increasing qubit number. Nevertheless, this reduced speedup by MPI on CFHS is balanced with the vast runtime improvement due to the locality-independent emulation approach. For the exemplary case of a 24-qubit random Hamiltonian with 100 terms and a locality of 24, we observe a runtime improvement of CFHS of a factor of $4.5$ compared to IQS when enabling MPI. In addition to the results presented here, we showcase the rescaled runtime results for the complete set of random Hamiltonians in Appendix~\ref{sec:appendix_rh_simtime_eval} with and without MPI.

\begin{figure*}[h!]
        \centering
        \includegraphics[width=0.45\textwidth]{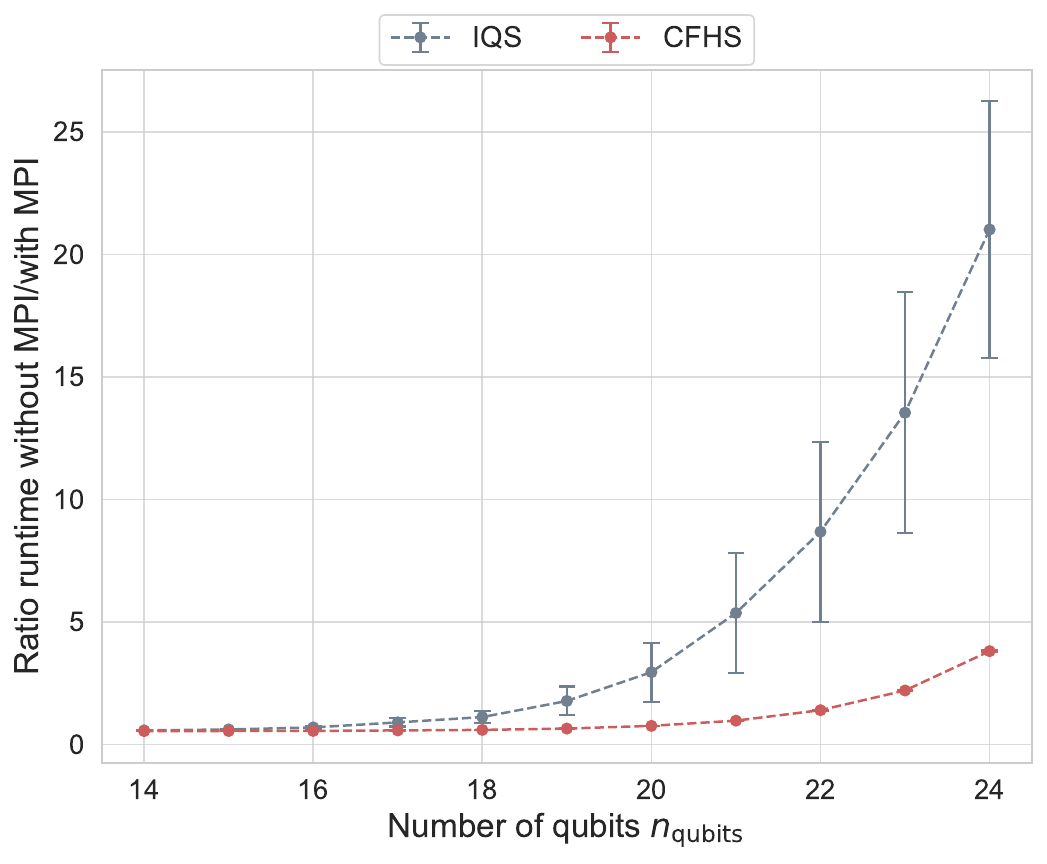}
        \caption{Demonstrating the impact when enabling MPI for IQS and CFHS using 128 MPI ranks. We present the ratio of the runtimes for random Hamiltonians with 100 terms each, with the presented mean and standard deviation of Hamiltonians for different localities $L$ and specific numbers of qubits $n_{\text{qubits}}$.}\label{fig:runtime_case24q_ratio_MPI}
\end{figure*}

\subsection{Evaluation on Chemistry Hamiltonians }
To demonstrate the performance improvement in CFHS simulation time for real-world applications in chemistry, we select molecular Hamiltonians that contain 14 to 24 qubits using the Jordan-Wigner representation for their qubit-based format from the Hamlib dataset~\cite{sawaya2024hamlib}, a dataset designed to make benchmarking studies more standardized and reproducible. For consistent evaluation, we set the number of Trotter steps to one, i.e., $M=1$ in Eq.~\eqref{eqn:trotter}. Increasing the number of Trotter steps $M$ would result in a linear growth in compilation and running time. The evaluated chemistry Hamiltonians are chosen in order to cover a broad range of parameter configurations, including the number of qubits, number of terms, and the one norm, which indicates the complexity of the system. In general, there are larger variations in the parameter configurations of the examined chemistry Hamiltonians, in contrast to the well-defined random Hamiltonians from before.
\subsubsection*{Compilation time} %
Similarly, for the random Hamiltonians, we track the compilation time, demonstrating that we are not offloading runtime towards compilation time for the evaluations on chemistry Hamiltonians. Our evaluation shows that the ratio of compilation time CFHS compared to IQS without MPI is $0.96 \pm 0.03$ and the ratio of compilation time CFHS compared to IQS with MPI is given by $0.98 \pm 0.08$. A detailed evaluation is added in Appendix~\ref{appendix:comptime_hamlib}, and exact compilation timing results can be found in Table~\ref{tab:scores} for completeness.
\subsubsection*{Simulation time} %
As with the random Hamiltonians, we are especially interested in revealing the scaling of simulation time with locality for the set of chemistry Hamiltonians, and we confirmed that there was no significant increase in compilation time from IQS to CFHS. Each of the considered chemistry Hamiltonians contains terms with a range of different localities, in contrast to the well-defined locality of the artificially generated random Hamiltonians. Thus, we consider the dependence of runtime on the mean locality $L_{\text{mean}}$ in Figure~\ref{fig:simtime_hamlib}. We compare the simulation time (without rescaling, thus corresponding to the wall-clock duration of the simulation) with the semi-logarithmic plot in Figure~\ref{fig:simtime_hamlib} to demonstrate how CFHS decreases the simulation time by a constant factor (see Figure~\ref{fig:hamlib_simtime_a}). This improvement in simulation time in CFHS remains evident when using MPI (see Figure~\ref{fig:hamlib_simtime_b}). When evaluating the runtime ratio comparing MPI usage or no MPI usage, we observe in Figure~\ref{fig:hamlib_ratio_runtime_improvement_MPI} that the improvement of MPI usage considering IQS and CFHS is more similar for both. Note that the Hamiltonians considered here are less similar and well-distributed in their parameter ranges than the random Hamiltonians considered before. Considering 24-qubit Hamiltonians as an exemplary case, for $O_3$, we obtain an improvement in runtime of a factor of ~$14.4$ for CFHS compared to IQS without MPI and a factor of $ 17.8$ with MPI. All timing results for our chemistry Hamiltonians can be found in Table~\ref{tab:scores}, together with their number of terms and details on their locality.
Furthermore, we present the rescaled simulation time as defined in Eq.~\eqref{eqn:workload} in Appendix~\ref{appendix:runtime_hamlib}. We observe a slight increase in the runtime of the IQS results with $L_{\text{mean}}$, both with and without MPI. The rescaled runtime of CFHS is approximately constant when varying the mean locality $L_{mean}$ for the chemistry Hamiltonians under consideration (see Figure~\ref{fig:workload_hamlib}) for Hamiltonians that contain more than 8 qubits. For Hamiltonians with a lower qubit number, other processes dominate the simulation time, and the performance improvement is less obvious.

\begin{figure*}[h!]
    \centering
    \begin{subfigure}[t]{0.5\textwidth}
        \centering
        \includegraphics[width=\textwidth]{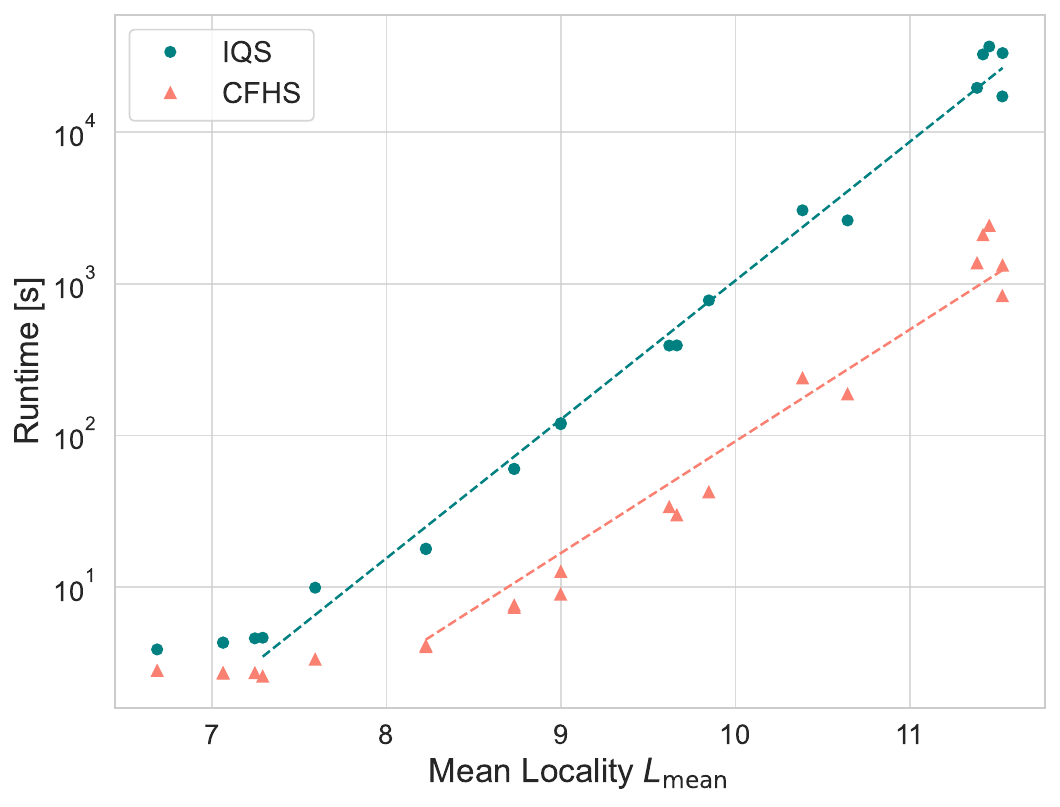}
        \caption{Simulation time [s] for selected chemistry Hamiltonians without MPI parallelization.}\label{fig:hamlib_simtime_a}
    \end{subfigure}%
    ~
    \begin{subfigure}[t]{0.5\textwidth}
        \centering
        \includegraphics[width=\textwidth]{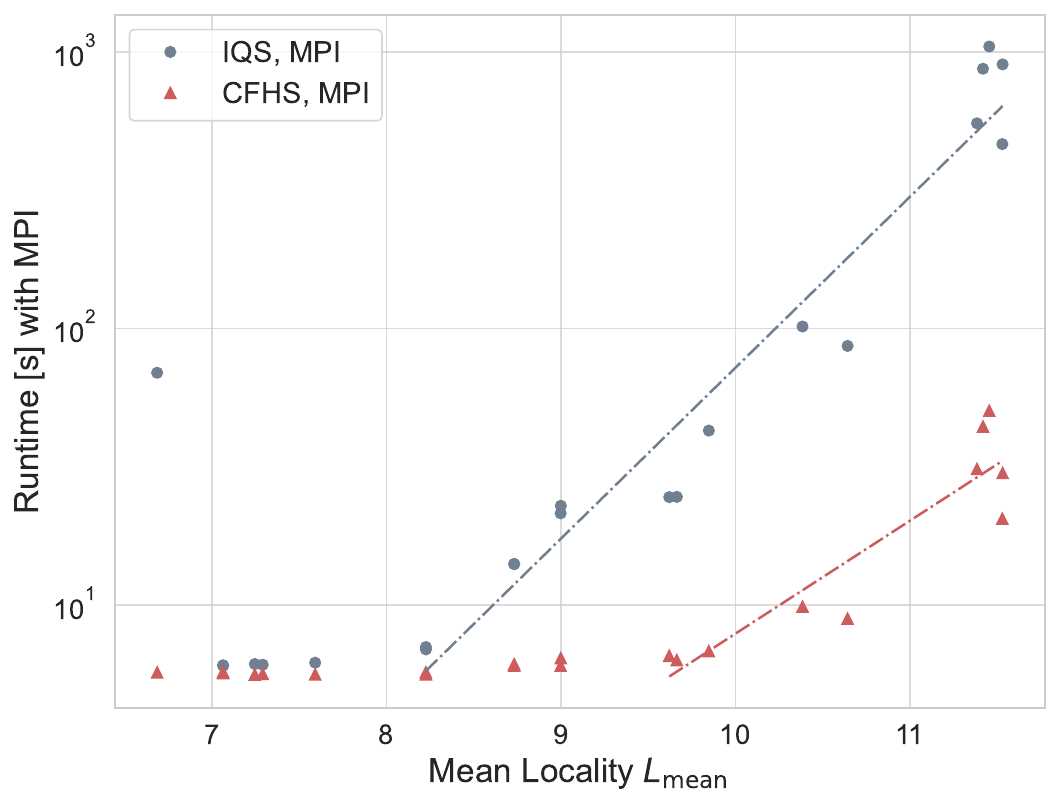}
        \caption{Simulation time [s] for selected chemistry Hamiltonians with MPI parallelization.}\label{fig:hamlib_simtime_b}
    \end{subfigure}
    \caption{Simulation time (runtime [s]) for the depicted chemistry Hamiltonians as in Table~\ref{tab:scores} ordered by the mean locality $L_{mean}$ of each Hamiltonian in Figure~\ref{fig:hamlib_simtime_a}. The results when enabling MPI parallelization are depicted in Figure~\ref{fig:hamlib_simtime_b}. Note that this is a semi-logarithmic plot, logarithmic in the y-axis and an approximate exponential scaling in the mean locality is indicated, the asymptotic scaling behavior sets in earlier for the runs without MPI as for the ones with  MPI.}\label{fig:simtime_hamlib}
\end{figure*}

\begin{figure}
    \centering
    \includegraphics[width=0.5\linewidth]{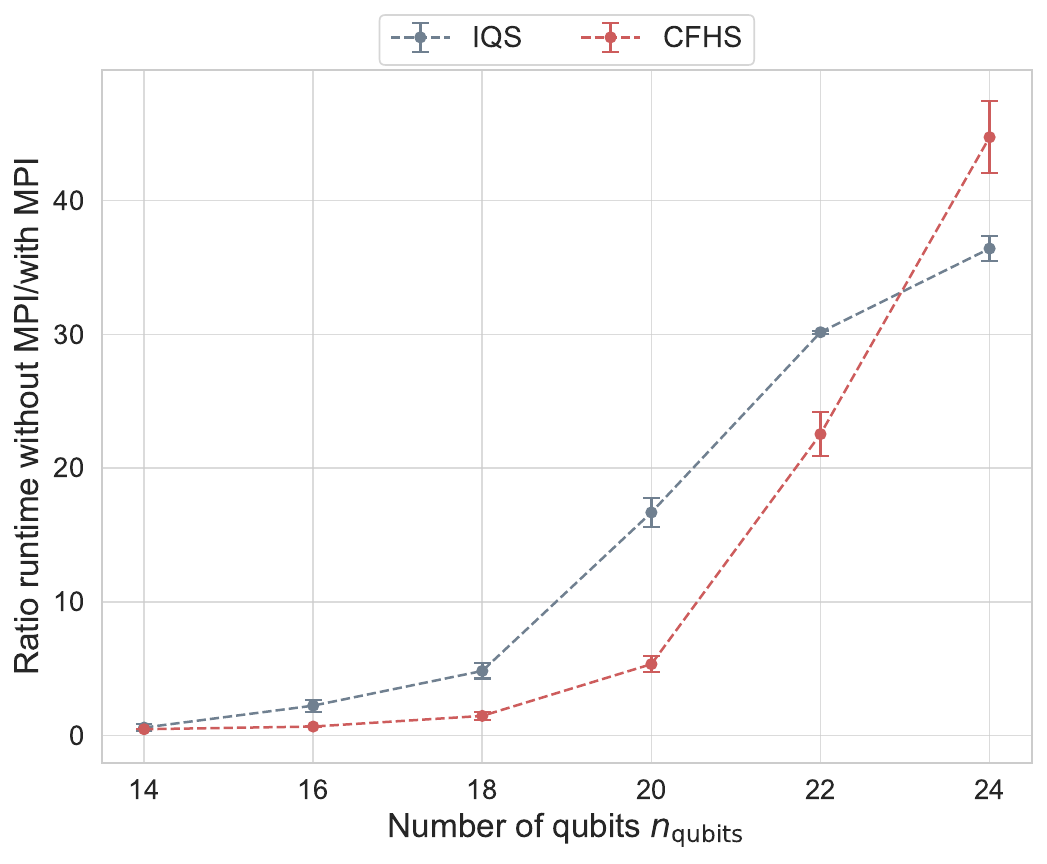}
    \caption{Impact when utilizing MPI for IQS and CFHS using 128 MPI ranks as before. Here, we show the results presenting the ratio of runtime of using no MPI and MPI for selected chemistry Hamiltonians from Hamlib as given in Table~\ref{tab:scores}. We show the mean and standard deviation taken over the Hamiltonians with the same number of qubits $n_{\text{qubits}}$ but different numbers of terms and localities, respectively.}
    \label{fig:hamlib_ratio_runtime_improvement_MPI}
\end{figure}

%% file: sec_conclusion.tex
In this work, we present a novel hybrid simulation approach for efficiently emulating multi-qubit rotations. We achieve this by tracking the state of the system with both a full state vector and a Pauli frame, with Clifford operations updating only the Pauli frame, which has a negligibly small cost. This is particularly impactful for Trotterized Hamiltonian simulation, considered to be a high-cost application of quantum computation in the chemistry or material science domain. Specifically, Trotterized Hamiltonian workloads lead to many high locality terms. With our Clifford fullstate hybrid simulator, multi-qubit rotations become as resource-intensive as one-qubit rotations by leveraging the Pauli frame as a lookup table. This has a significant impact on both the time required for quantum simulation, as only a fraction of operations impacts the state vector, and the communication overhead in distributed settings, as the Pauli frame is compact enough to be shared cheaply among nodes, in contrast to portions of the state vectors. 
To demonstrate this performance improvement, we evaluate different workloads consisting of artificially generated random Hamiltonians with a specified number of terms and locality for a given number of qubits involved. When evaluating the simulation time in terms of locality, we observe a sizeable improvement in CFHS over the original IQS variant, as the simulation time for CFHS is independent of the locality of the Hamiltonian terms. This is in contrast to the original IQS implementation, where the runtime depends linearly on the Hamiltonian locality. For our evaluated chemistry Hamiltonians with 24 qubits, a speedup of a factor of $18 \pm 4$ is achieved of CFHS over IQS without MPI, and a speedup of a factor $22 \pm 4$ with MPI. Another way to look at this speedup is by considering the ratio of non-Clifford \emph{vs} Clifford gates: Only the former ones need to be applied by the CFHS on the state vector (by far the most expensive update), thus reducing the overall simulation time accordingly. This performance improvement is clearly demonstrated by artificial random Hamiltonians, each of which has only terms with a specified locality, and verified by chemistry Hamiltonians that contain terms with a range of localities. For the chemistry Hamiltonians, the performance improvement from IQS to CFHS is slightly reduced when MPI is enabled, but becomes more pronounced as the maximum locality increases. Considering the compilation time, there is none to no significant overhead comparing CFHS to the original IQS with and without MPI.
Thus, we demonstrated the impact of improving simulation time by benchmarking a set of various workloads using our novel Clifford Fullstate Hybrid simulator, which comes with no difference in compilation time but a significant improvement in simulation time, by removing the linear scaling in terms of locality that is provided using IQS. Thus, our new hybrid simulation approach presents a vast performance improvement in simulation time, especially for Hamiltonians with high locality. This can be of high importance not only for chemistry simulation, but also in high-energy physics, for example, considering the Schwinger model. Our new hybrid approach could be beneficial to other types of simulators, and in the future, one could investigate a sparse representation of the fullstate vector in order to completely remove the overhead for simulating Clifford circuits.

%% file: sec_appendix.tex
\section{Efficiently implementing Clifford unitaries with multi-qubit rotations}\label{sec:appendix_constralg}
Any Pauli frame $F$ or its corresponding Clifford unitary $U$ can be implemented or inverted using at most~$\mc O(n)$ multi-qubit rotations. In order to demonstrate this, as before, we define our Pauli Frame~$F$ as:
\begin{align}
F = \begin{pmatrix}
\eff Z_0 & \eff X_0 \\
\vdots & \vdots \\
\eff Z_{n-1} & \eff X_{n-1}
\end{pmatrix} \, .
\end{align}
 This represents the action of the respective Clifford unitary $U$ on the Paulis $Z_j$ and $X_j$ with $0 \leq j \leq (n-1)$ and without considering signs, resulting in an $n \times 2$ matrix. Note that we do not consider the phase. 
\begin{theorem}
    A Clifford unitary $U \in \mathcal{C}_n$ can be implemented using at most~$\mc O(n)$ multi-qubit rotations of the form $R(\theta) = e^{-i\theta/2 P}$ with Paulis $P \in \mathcal{P}_n$. 
\end{theorem}

\begin{proof}
   In order to get the sequence of rotations \textit{implementing} the Clifford unitary, we reverse the rotation sequence and negate each rotation (i.e., $e^{-i\theta/2 P} \rightarrow e^{+i\theta/2 P}$). The idea is now to map the Pauli Frame $F$ representing $U$ back to representing the initial Pauli Frame~$F_{origin}$:
    \begin{align}
    \begin{pmatrix}
\eff Z_0 & \eff X_0 \\
\vdots & \vdots \\
\eff Z_{n-1} & \eff X_{n-1}
\end{pmatrix} \, 
\mapsto 
\begin{pmatrix}
 Z_0 &  X_0 \\
\vdots & \vdots \\
 Z_{n-1} &  X_{n-1}
\end{pmatrix} =: F_{origin} \, ,
    \end{align} 
    while implementing $U^\dagger$ in the form of a sequence of Pauli rotations. If it holds that $U \in \mathcal{P}_n$, it is straightforward to implement a rotation as $U^\dagger=P \in \mathcal{P}_n$, then $e^{-i\pi/2 P}$ forms the corresponding multi-qubit rotation. If we now consider any $U^\dagger$ in the group $\mathcal{C}_n \not \, \,\,  \mathcal{P}_n$. As it holds that 
    \begin{align}
    \mathcal{C}_n \not \, \,\,  \mathcal{P}_n \cong  \text{Sp}(2n, \mathbb{F}_2) \equiv \text{Sp}(2n) \, ,
    \end{align}
    i.e., there is a bijective map to the symplectic group on $\mathbb{F}_2^{2n}$, which is generated by symplectic transvections~\cite{koenig2014efficiently, o1978symplectic}. Every symplectic transvection can be written of the form  $\frac{1}{\sqrt{2}}(\mathbb{I} - i P) = e^{-i\pi/4 P}$ for some Pauli $P$~\cite{pllaha2020weyl}. As given in~\cite{koenig2014efficiently}, any element in the group $\text{Sp}(2n)$ can be decomposed into at most $4n$ symplectic transvections. The multi-qubit Pauli rotation $e^{-i\pi/4 P}$ is a symplectic transvection and thus, any element in $\text{Sp}(2n)$ can be inverted using $\mc O(n)$ multi-qubit rotations. For the implementation of $U$ using $\mc O(n)$ multi-qubit rotations, we use the approach of reversing the rotation sequence as described in the beginning, which leads to no overhead in the number of rotations applied. 
\end{proof}

For completeness, we also present the constructive algorithm for implementing or inverting a Pauli frame using $\mathcal{O}(n)$ Pauli rotations. Specifically, we consider the inversion of the Pauli Frame through a sequence of Pauli rotations and, thus, we again express the following transformation in the form of $n$ Pauli rotations:
   \begin{align}
    \begin{pmatrix}
\eff Z_0 & \eff X_0 \\
\vdots & \vdots \\
\eff Z_{n-1} & \eff X_{n-1}
\end{pmatrix} \, 
\mapsto 
\begin{pmatrix}
 Z_0 &  X_0 \\
\vdots & \vdots \\
 Z_{n-1} &  X_{n-1}
\end{pmatrix} \, .
\end{align} 
As noted before, the reversal is achieved through implementation, and in this case, the sequence of Pauli rotations is only reversed, with each rotation negated. For the following constructive algorithm, the following observation on the commutativity and anti-commutativity of Pauli rotations is relevant. We first define a Clifford Pauli rotation. 
\begin{definition}[Clifford Pauli rotation]
    A Clifford Pauli rotation is a rotation $e^{-i\theta Q}$ with $ Q \in \mathcal{P}_n$ such that:
    \begin{align}
        e^{-i\theta Q} \,  P \, e^{+i\theta Q} \in \mathcal{P}_n \, \forall P \in \mathcal{P}_n \, .
    \end{align}
    As for $P, Q$ we either have $[P, Q]=0$ or $\{P, Q\}=0$ , in the former case we have $e^{-i\theta Q} \,  P \, e^{+i\theta Q} =P$ and $e^{-i\theta Q} \,  P \, e^{+i\theta Q}= iPQ \sin(2\theta) + P \cos(2\theta)$ in the anti-commuting case. As $PQ\in \mathcal{P}_n$ or $P\in \mathcal{P}_n$, we must have that $\theta =\frac{\pi}{4}n $ with $n\in \mathbb{Z}$.
\end{definition}
\begin{note}
    Consider a Pauli $P \in \mathcal{P}_n$ and a Clifford Pauli rotation $R(Q) =e^{-i \frac{\pi}{4}Q}$ generated by a Pauli $Q \in \mathcal{P}_n$. If one conjugates the Pauli $P$ with the rotation $R(Q)$, we get:
    \begin{align}
        R(Q) \, P \, (R(Q))^{\dagger} = \begin{cases}
            P \, , \, &\text{ if } [P,Q]=0 \, , \\
            R (Q)^2 P = -iQ \, P \, , \, &\text{ if } \{P,Q \}=0 \, .
        \end{cases}
    \end{align}
\end{note}
As before, we aim to construct a unitary $U$ of Pauli rotations that has the following effect:
\begin{align}
    U P_i U^{\dagger} = Z_i \, \, \forall \, i \, .
\end{align}
We construct the unitary $U$ from $n$ Pauli rotations of the form $R(Q) =e^{-i \frac{\pi}{4}Q}$.
\newpage 
\subsubsection*{Constructive Algorithm: Compute the Clifford Pauli rotations for Pauli frame inversion}
Given a Pauli frame $F$ on $N$ qubits $\{q_j\}_{0\leq j<N}$:

\begin{enumerate}

\item Let $F_0 \gets F$.

\item For each $0 \leq j < N$ in order for any ordering of the qubits, $\{q_j\}_{0\leq j<N}$:

\begin{enumerate}

\item Find any index $i$ of $F_j $ such that $(F_j)_i = (\eff Z_i, \eff X_i)$ and $\eff Z_i$ has support on $q_j$.\footnote{Such an $i$ always exists for any $q$ as we can show by expanding the following in any frame $F = \{(\eff Z_i, \eff X_i)\}_i$:

\begin{align}
1= \lambda(Z_q, X_q) = \sum_i \left( \lambda(Z_q, \eff Z_i) \lambda(\eff X_i, X_q) \oplus \lambda(Z_q, \eff X_i) \lambda(\eff Z_i, X_q) \right).
\end{align}
For this to be true, at least one term $i$ of the sum must be non-zero, in which case $\lambda(Z_q, \eff Z_i)=1$ or $\lambda(\eff Z_i, X_q)=1$ . }

\item The Pauli frame element $\eff Z_i$ has support on $q_j$ of either $X$, $Y$ or $Z$. Let this support be $\sigma$ and $\tilde \sigma$ be any other single-qubit Pauli operator support such that $\lambda(\sigma_{q_j}, \tilde \sigma_{q_j}) = 1$ and thus  $\lambda(\eff Z_i, \tilde \sigma_{q_j}) = 1$. Then let 

\begin{align}
F'_{j +1} \gets \exp\left(-i\frac{\pi}{4} Q\right ) F_j \exp\left(i\frac{\pi}{4} Q\right ),
\end{align}
 
where $Q = i \tilde \sigma_{q_j} * \eff Z_i$. Note this implies $(F'_{j+1})_i = (\tilde \sigma_{q_j}, \eff X_i')$ for some Pauli operator $\eff X_i'$.

\item Let $ \doubletilde \sigma$ be the Pauli operator support which is not $\tilde \sigma$ and not that of $\eff X_i'$'s support on $q_j$.\footnote{Note again, such support must exist by the frame properties of $F'_{j+1}$ as $\eff X_i'$ must anti-commute with the single qubit $\tilde \sigma_{q_j}$. For example, let $\tilde \sigma = Z$ and the support of $\eff X$ on $q_j$ be $X$. Then $\doubletilde \sigma = Y$.  } Then let

\begin{align}
F_{j +1} \gets \exp\left(-i\frac{\pi}{4} Q'\right ) F'_{j + 1} \exp\left(i\frac{\pi}{4} Q'\right ),
\end{align}

where $Q' = i \doubletilde \sigma_{q_j} * \eff X_i'$. Note that $\lambda(Q', \eff X_i') = \lambda(\doubletilde \sigma_{q_j}, \eff X_i') = 1$ and $\lambda(Q', \tilde \sigma_{q_j}) = \lambda(\doubletilde \sigma_{q_j}, \tilde \sigma_{q_j}) \oplus \lambda(\eff X_i', \tilde \sigma_{q_j}) = 1 \oplus 1= 0$. This implies that $(F_{j+1})_i = (\tilde \sigma_{q_j}, \doubletilde \sigma_{q_j})$, and we have reduced the position $i$ of Pauli frame $F_{j+ 1}$ to single qubit support on $q_j$. By the frame properties, no other positions of  $F_{j+ 1}$ have support on $q_j$, and thus the position~$i$ will be unaffected by later transformations.

\end{enumerate}

\item Reduce $F_N$ by non-entangling transformations using IQS's native gate operations to the origin frame $F_{origin}$. We easily compute the non-entangling transformations (single-qubit operations and applications of SWAP gates) to transform back to the original frame:
    \begin{enumerate}
        \item For each position $i$, map the symplectic pair to its original form $(Z_q, X_q)_i$. In total, there are only 24 possible single-qubit Clifford gates that map the symplectic pair at position $i$ for qubit $q$ back to the original form: $(\sigma_q, \sigma'_q)_i \rightarrow (Z_q, X_q)_i$ consisting of at most three single-qubit Pauli rotations.
        \item Map each symplectic pair corresponding to qubit $q$ to its correct position $i$ using SWAP gates which is a native operation in IQS done via index relabeling\footnote{A SWAP gate can be also expressed via multi-qubit rotations: a SWAP for qubits $0$ and $1$ can by implemented by $\pi /4$ rotations about the axis $X_0X_1$, $Y_0 Y_1$ and $Z_0 Z_1$.}.
    \end{enumerate}
    After this step, we have recovered the frame in its original form $F_{origin}$. 
\end{enumerate}

\newpage 
\section{Amplitude update rule for multi-qubit rotations}
\label{sec:appendix_update_rule_multiqubit_rotation}

In the notation of section~\ref{sec:fullstatesim}, we discuss the implementation of rotations generated by the $n$-qubit Pauli $P$ on an arbitrary state $\ket{\psi} = \sum_k \alpha_k \ket{k}$. Since the operation is linear, we can express the update rule by using the transformation of the computational basis $\ket{k}$:

\begin{align}
    R_P(\theta) \ket{\psi} =& \sum_k \alpha_k R_P(\theta) \ket{k} \\
    =& \sum_k \alpha_k \left[ \cos{(\theta/2)} \ket{k} -i \sin{(\theta/2)}
    e^{i \varphi_P(k)} \ket{f_P(k)} \right]\, .
\end{align}
Now there are two situations:
When $P$ is a product of only $I$ and $Z$ single-qubit Paulis, then $f_P(k)=k$, and each probability amplitude $\alpha_k$ is independently updated with the phase:

\begin{equation}
\left[ \cos(\theta/2) - i \sin(\theta/2) e^{i \varphi_P(k)}\right] 
= \begin{cases}
        e^{-i \theta/2 } \, , & \text{if } \mod(|m_Z(P)) \, \odot \, k|, 2) = 1 \\
        e^{ i \theta/2 } \, , & \text{otherwise} \, . 
    \end{cases} 
\end{equation}
where we have used the fact that in this scenario $e^{i \varphi_P(k)} = (-1)^{|m_Z(P)) \, \odot \, k|}$.

In the general case of $P$ including $X$ or $Y$ factors, one still has $f_P(f_P(k))=k$ and the amplitudes can be updated in pairs according to:

\begin{equation}
    \begin{pmatrix}
        \alpha_k^\prime \\
        \alpha_{f_P(k)}^\prime
    \end{pmatrix}
=
    \begin{pmatrix}
    \cos(\theta/2) & -i \sin(\theta/2) e^{i \varphi_P(f_P(k))} \\
    -i \sin(\theta/2) e^{i \varphi_P(k)} & \cos(\theta/2)
    \end{pmatrix} 
    \begin{pmatrix}
        \alpha_k \\
        \alpha_{f_P(k)}
    \end{pmatrix}
\end{equation}
where the primed amplitudes $\alpha_k^\prime$ and $\alpha_{f_P(k)}^\prime$ indicate their value after the update.

\newpage

\section{Complete workflow of the hybrid Clifford fullstate simulation}
In Figure~\ref{fig:gateupdateflow}, we describe the workflow of the hybrid fullstate simulator in detail. In the case that $G$ is Clifford, the frame $F$ is updated, and otherwise the fullstate $\Phi$. Thereby, we not only distinguish whether $G$ is a Clifford gate or not, but further distinguish whether $G$ is a state preparation, rotation, or measurement in the case $G$ is not a Clifford gate. If $G$ forms a non-Clifford (multi-qubit) rotation, the rotation axis is updated based on the up-to-date Pauli Frame. 
\begin{figure}[h!]
    \centering
    \scalebox{0.75}{
    \input{figures/CFHS_gate_flow}
    }
    \caption{Workflow of the gate-updates in the Clifford fullstate hybrid Simulator. The flow diagram describes a simulator update by a gate $G$. First, one needs to check whether $G$ is a Clifford gate or not. In case the gate $G$ is Clifford, it is used to update the Pauli Frame $F$. Otherwise, if $G$ is non-Clifford, the gate $G$ is used to update the fullstate $\Phi$. Therefore, the process differs depending on whether $G$ represents a state preparation, rotation, or measurement. Altogether, the Pauli Frame $F$ is returned along with the fullstate $\Phi$. }
    \label{fig:gateupdateflow}
\end{figure}
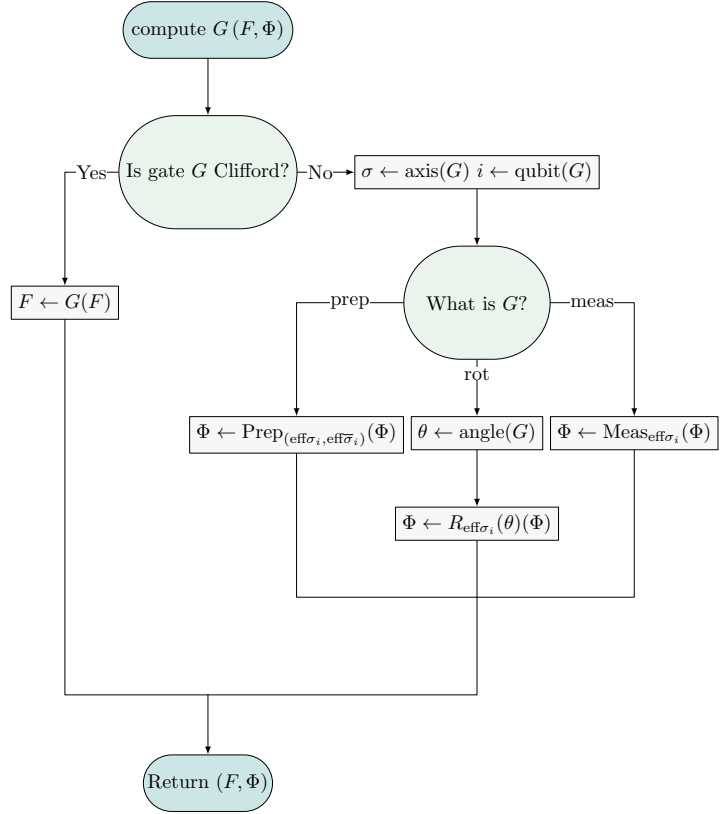
\newpage
\section{Used Hardware for Benchmarking and Details on Parallelization with Message Passing Interface}

We evaluated the compilation and simulation times for each Hamiltonian workload and compared the results obtained with a message-passing interface (MPI) for parallelization with those without it. The MPI runs used 4 nodes, 64 MPI tasks, allocating 4 CPUs (cores) per task. 
Absolute timings are, of course, dependent on the processor, memory, and simultaneous occupancy of the cluster. For the presented benchmarking evaluation, each node in the partition used has the following specifications: dual-socket Xeon Platinum 8592+, 2x64 cores Emerald Rapids, 1024GB DDR5, dual-rail OPA.

\section{Random Hamiltonian evaluation results}
\subsection{Evaluation of compilation time for random Hamiltonians}\label{sec:appendix_comptime_eval}
In Figure~\ref{fig:comptime_avg_loc_100terms}, we show the linear increase of the compilation time in terms of locality~$L$. Furthermore, we observe that CFHS does not result in a noticeable increase in compilation time overhead compared to the initial IQS version. With and without MPI usage, the compilation time of IQS is approximately the same as that of CFHS. Note that many compilation times are less than 10 seconds, and fluctuations in the ratio arise due to small differences and limited tracking accuracy. To highlight, the distribution displayed in Figure~\ref{fig:comptime_ratio_rh} is centered around one.
\begin{figure*}[h!]
    \centering
    \begin{subfigure}[t]{0.5\textwidth}
        \centering
        \includegraphics[width=\textwidth]{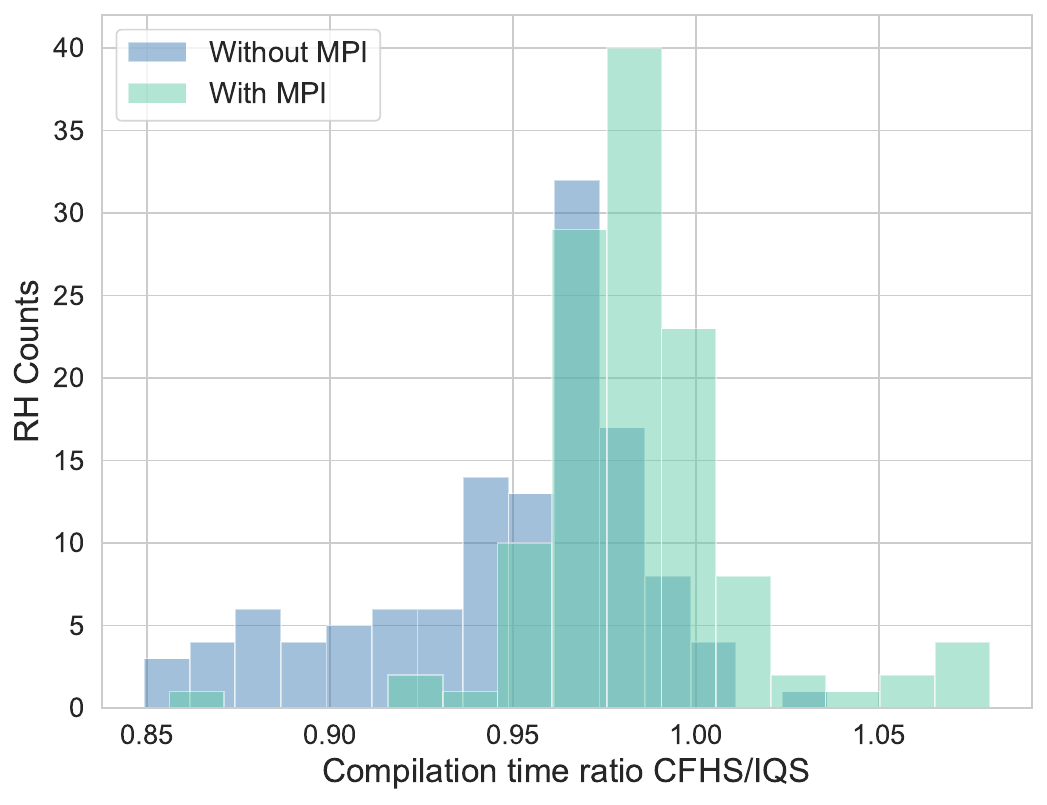}
        \caption{Histogram of the compilation time ratio distribution for CFHS/IQS.}\label{fig:comptime_ratio_rha}
    \end{subfigure}%
    ~
 \begin{subfigure}[t]{0.535\textwidth}
        \centering
        \includegraphics[width=\textwidth]{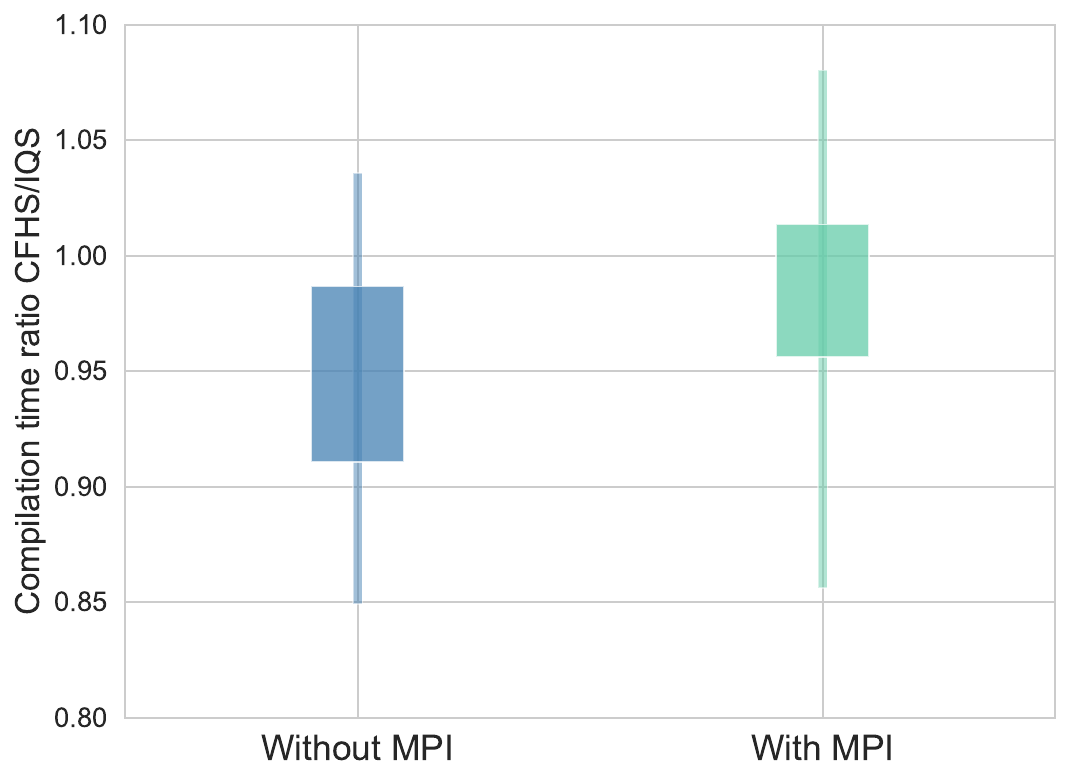}
        \caption{Mean and standard deviation of the compilation time ratio distribution for CFHS/IQS.}\label{fig:comptime_ratio_rhb}
    \end{subfigure}%
    \caption{Evaluation of the compilation time ratio distribution for CFHS/IQS that is measured in seconds for all random Hamiltonians (RH) with 100 terms each. In Figure~\ref{fig:comptime_ratio_rhb}, the wider bar denotes the mean of the ratio plus or minus one standard deviation. The thinner bar showcases the minimum and maximum value of the ratio. We compare the case without MPI to the case with MPI. }\label{fig:comptime_ratio_rh}
\end{figure*}

\begin{figure}[h!]
    \centering
    \includegraphics[width=0.55\textwidth]{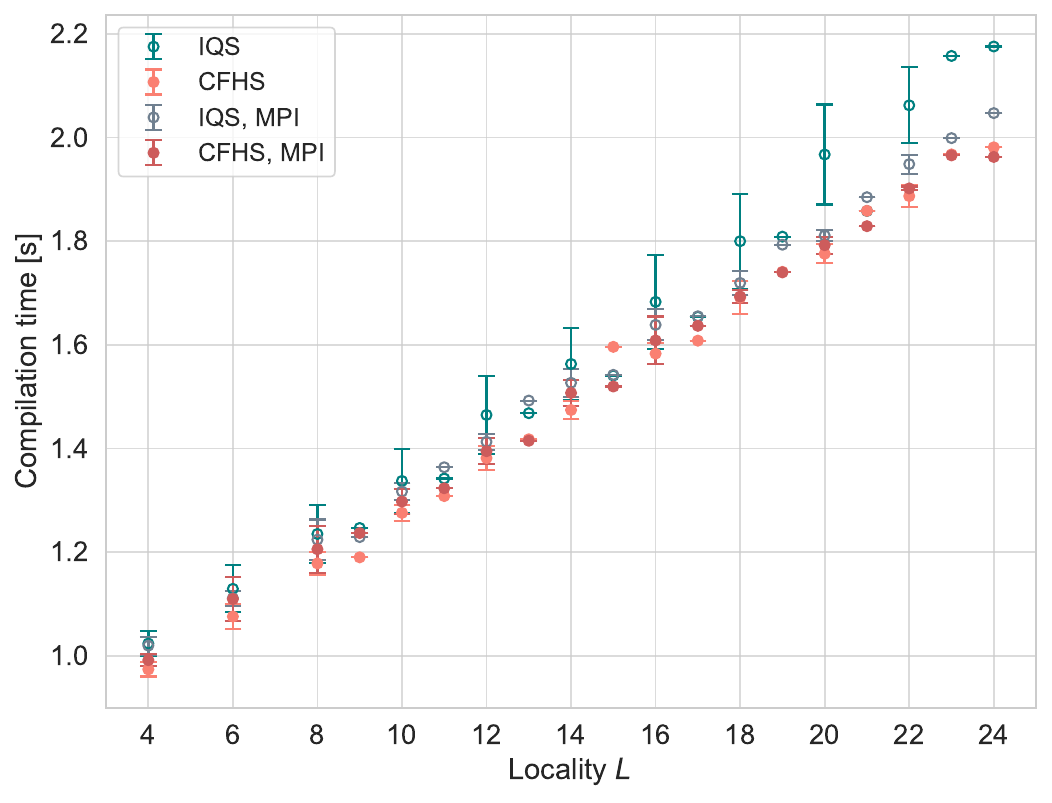}
    \caption{Increase of the compilation time [s] for random Hamiltonians each with 100 terms, averaged over the respective locality of the Hamiltonian for given numbers of qubits. We compare both the compilation time with and without MPI parallelization, as well as the compilation time results obtained via IQS and the CFHS simulator. We showcase the mean and standard deviation as error bars.}
    \label{fig:comptime_avg_loc_100terms}
\end{figure}
    \label{fig:comptime_avg_nq_100terms}

\subsection{Evaluation of simulation time for random Hamiltonians}\label{sec:appendix_rh_simtime_eval}
Accompanying the simulation time plots presented above and in order to better visualize the scaling of the simulation time being independent of the locality of the respective multi-qubit Pauli rotation, we define the \textit{rescaled simulation time} as follows:
\begin{align}
    \text{Rescaled Runtime [s]} := \frac{\text{Runtime [s]}}{n_{\mathrm{terms}}\cdot 2^{n_{\mathrm{qubits}}}} \, . \label{eqn:workload}
\end{align}
We note that the rescaled running time of CFHS is expected to be independent of the locality of the respective Hamiltonian. The numerical results for our evaluations of random Hamiltonians with distinct parameter configurations are presented in Figure~\ref{fig:sim_workload}. While Figure~\ref{fig:sim_workloada} represents the results without using MPI parallelization, we observe for Hamiltonians with a higher locality an increasing improvement between the results obtained with CFHS compared to the IQS results with the locality $L$. This difference is less obvious when using MPI parallelization in Figure~\ref{fig:sim_workload_mpi}. Note the difference in scale of a magnitude in the zoomed-in plots. In both Figures, we observe a regime marked in gray in the Figures~\ref{fig:sim_workloada} and~\ref{fig:sim_workload_mpi} due to rescaling as in Eq.~\eqref{eqn:workload}, where the asymptotic scaling is not yet reached for Hamiltonians with small qubit numbers $n_{\text{qubits}}$.

\begin{figure*}[h!]
    \centering
    \begin{subfigure}[t]{0.5\textwidth}
        \centering
        \includegraphics[width=\textwidth]{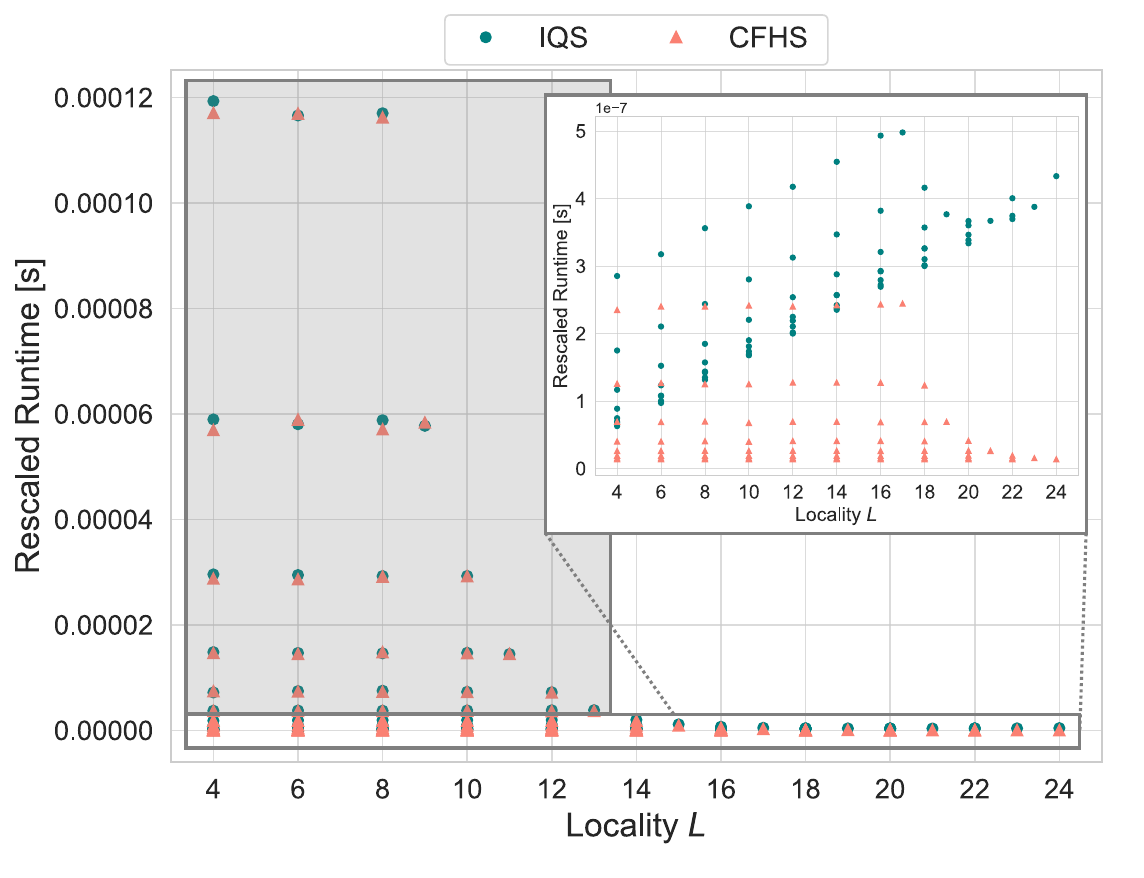}
        \caption{Rescaled simulation time [s] without using MPI for random Hamiltonians. For better visibility, we present the results for Hamiltonians with $n_{\mathrm{qubits}} > 16$ in detail.}\label{fig:sim_workloada}
    \end{subfigure}%
    ~
    \begin{subfigure}[t]{0.491\textwidth}
        \centering
        \includegraphics[width=\textwidth]{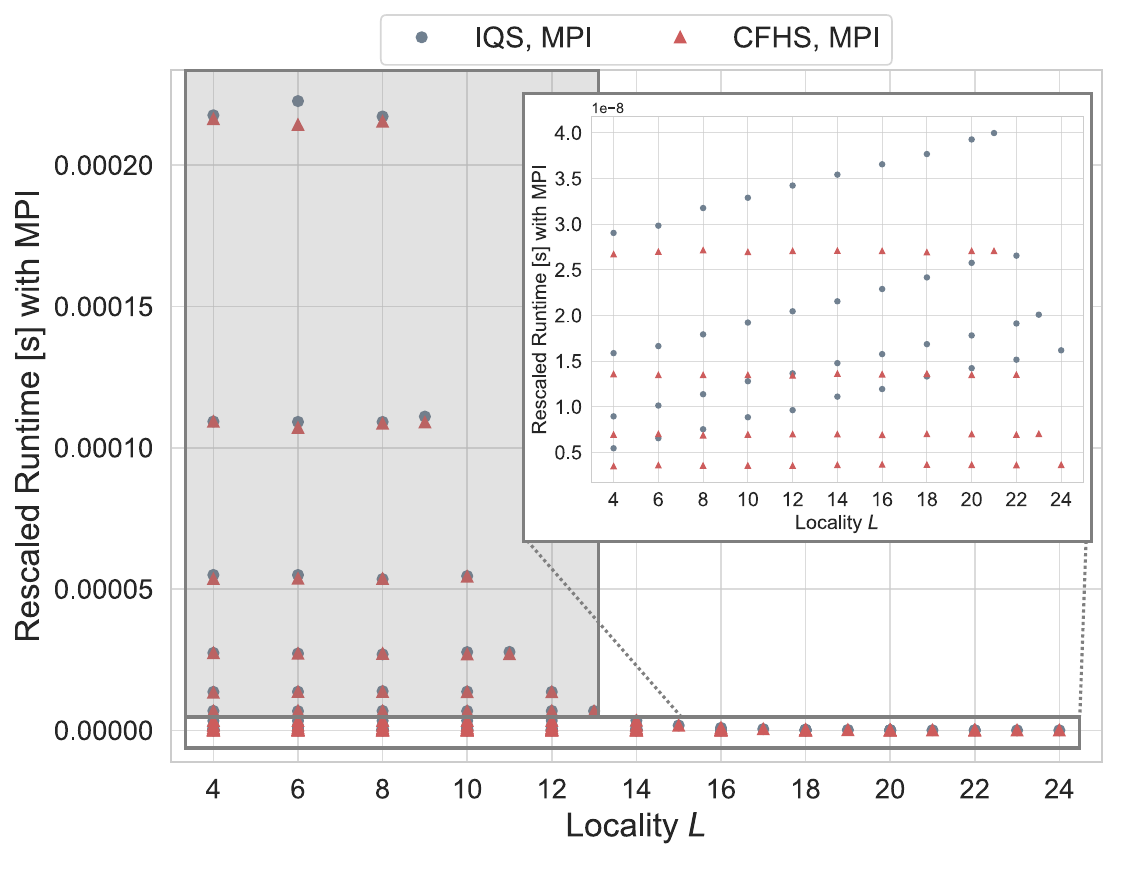}
        \caption{Rescaled simulation time [s] using MPI for random Hamiltonians. For better visibility, we present the results for Hamiltonians with $n_{\mathrm{qubits}} > 20$ in detail.}\label{fig:sim_workload_mpi}
    \end{subfigure}
    \caption{Rescaled runtime [s] as defined in Eq.~\eqref{eqn:workload} per random Hamiltonian with 100 terms and a defined locality $L$ ranging from 4 to 24, comparing results without (Figure~\ref{fig:sim_workloada}) and with MPI parallelization (Figure~\ref{fig:sim_workload_mpi}). We mark a regime in grey where the asymptotic behavior of the simulator is not yet achieved for small qubit numbers $n_{\text{qubits}}$. This is more apparent for low locality $L$ since this is the only place where small qubit numbers $n_{\text{qubits}}$ can contribute. %
    } \label{fig:sim_workload}
\end{figure*}

\section{Chemistry Hamiltonian evaluation results}
\subsection{Evaluation of compilation time for chemistry Hamiltonians}\label{appendix:comptime_hamlib}
\begin{figure*}[h!]
    \centering
    \begin{subfigure}[t]{0.5\textwidth}
        \centering
        \includegraphics[width=\textwidth]{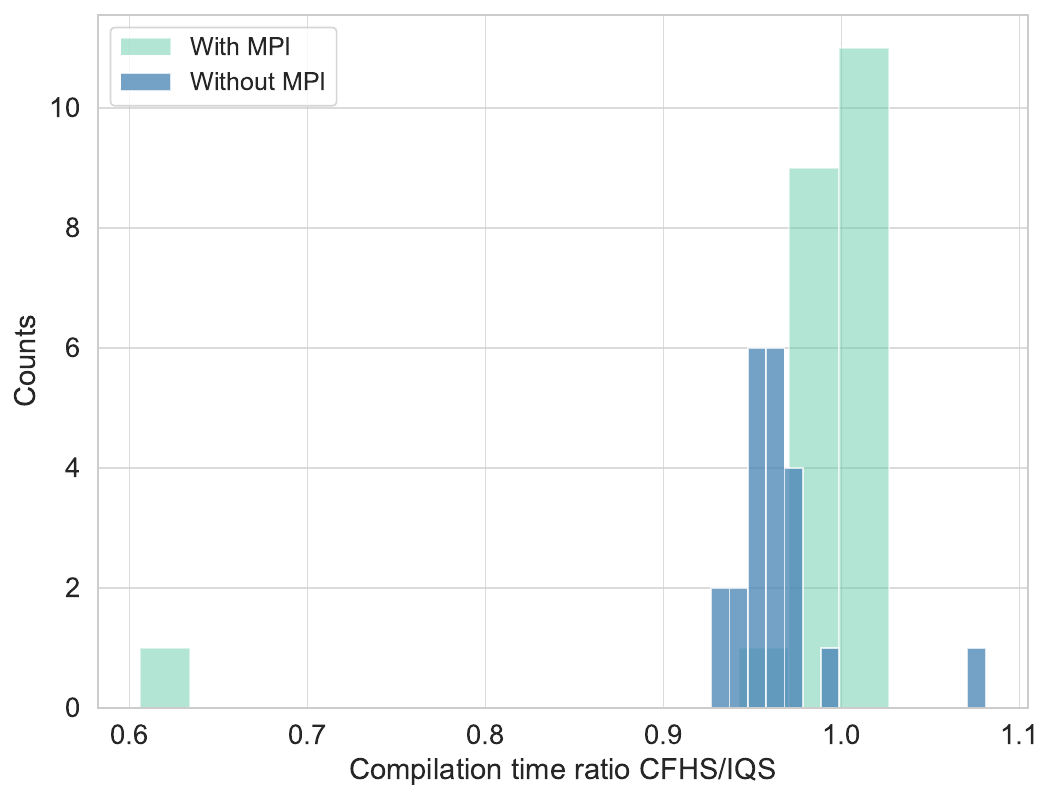}
        \caption{Histogram of the compilation time ratio distribution for CFHS/IQS.}\label{fig:ratio_comptime_hamliba}
    \end{subfigure}%
    ~
 \begin{subfigure}[t]{0.52\textwidth}
        \centering
        \includegraphics[width=\textwidth]{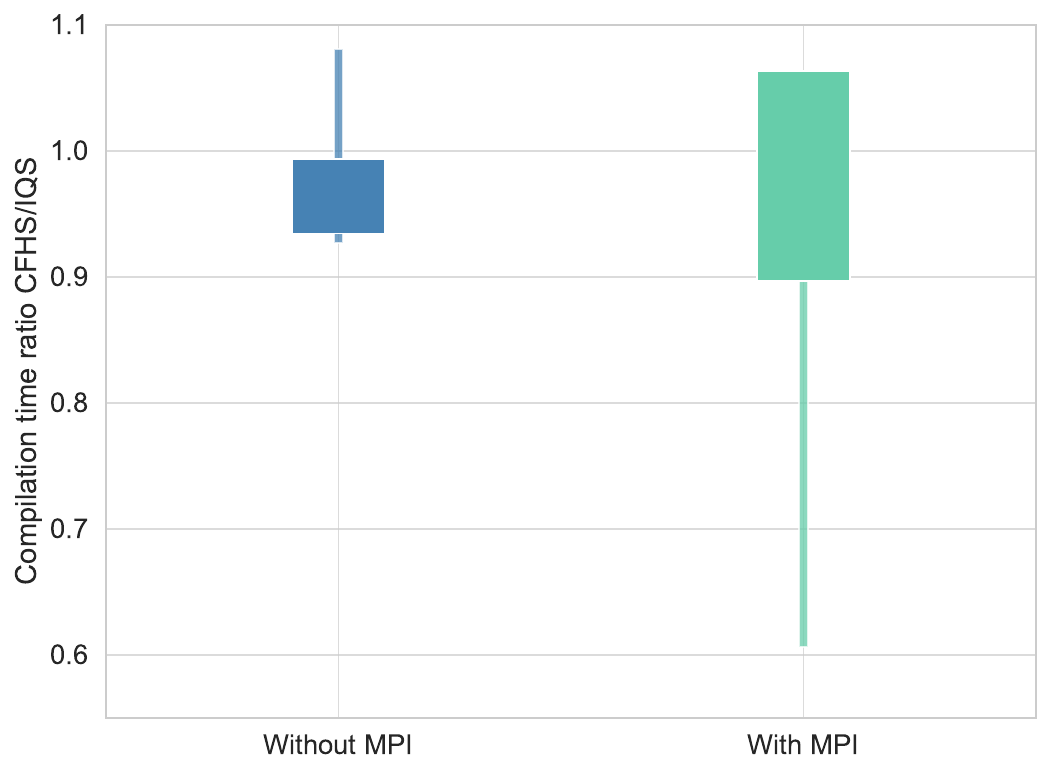}
        \caption{Mean and standard deviation of the compilation time ratio distribution for CFHS/IQS.}\label{fig:ratio_comptime_hamlibb}
    \end{subfigure}%
    \caption{Ratio of the compilation time [s] of CFHS/IQS for the evaluated chemistry Hamiltonians (CH) from Hamlib. We compare the compilation time results without MPI to the case with MPI. In Figure~\ref{fig:ratio_comptime_hamlibb}, the wider bar denotes the mean of the ratio plus or minus a standard deviation, and the thinner bar showcases the minimum and maximum value of the respective ratio value. Figure~\ref{fig:ratio_comptime_hamliba} visualizes the compilation time ratio of CFHS/IQS as a histogram. There is one outlier present corresponding to a Hamiltonian with the lowest number of terms and a compilation time below 5 seconds.}\label{fig:ratio_comptime_hamlib}
\end{figure*}
For the chemistry Hamiltonians, we also track the compilation time to demonstrate that no workload is offloaded from the simulation time. The compilation time results for the set of Hamiltonians as given in Table~\ref{tab:scores} are presented in Figure~\ref{fig:ratio_comptime_hamlib}. As expected, the distribution is again centered around one. There are a few outliers, corresponding to Hamiltonians with a low number of terms and qubits that have compilation times below 5 seconds. Exact timing results can be found in Table~\ref{tab:scores}. Small differences in compilation time ratios become less prominent for Hamiltonians corresponding to larger qubit numbers and localities.
\subsection{Evaluation of simulation time for chemistry Hamiltonians}\label{appendix:runtime_hamlib}
As for the random Hamiltonians, we present the \textit{rescaled runtime} as defined in Eq.~\eqref{eqn:workload}. As before, we highlight the regime where the asymptotic scaling has not yet set in, and zoom in on parts of the plot for better visibility. In general, the asymptotic regime sets in at a higher mean locality $L_{\text{mean}}$ for the MPI results. The gap of IQS and CFHS is nicely visible in Figure~\ref{fig:workload_hamliba}, where an approximately constant dependence of the CFHS rescaled runtime on the mean locality can be observed, whereas the IQS rescaled runtime increases slightly (approximately in a linear way) with the mean locality. For the MPI results displayed in Figure~\ref{fig:workload_hamlibb}, the improvement in the presented regime corresponds approximately to a constant factor for a mean locality between 10 and 11, but note the difference in magnitude comparing results with and without MPI. 

\begin{figure*}[h!]
    \centering
    \begin{subfigure}[t]{0.496\textwidth}
        \centering
        \includegraphics[width=\textwidth]{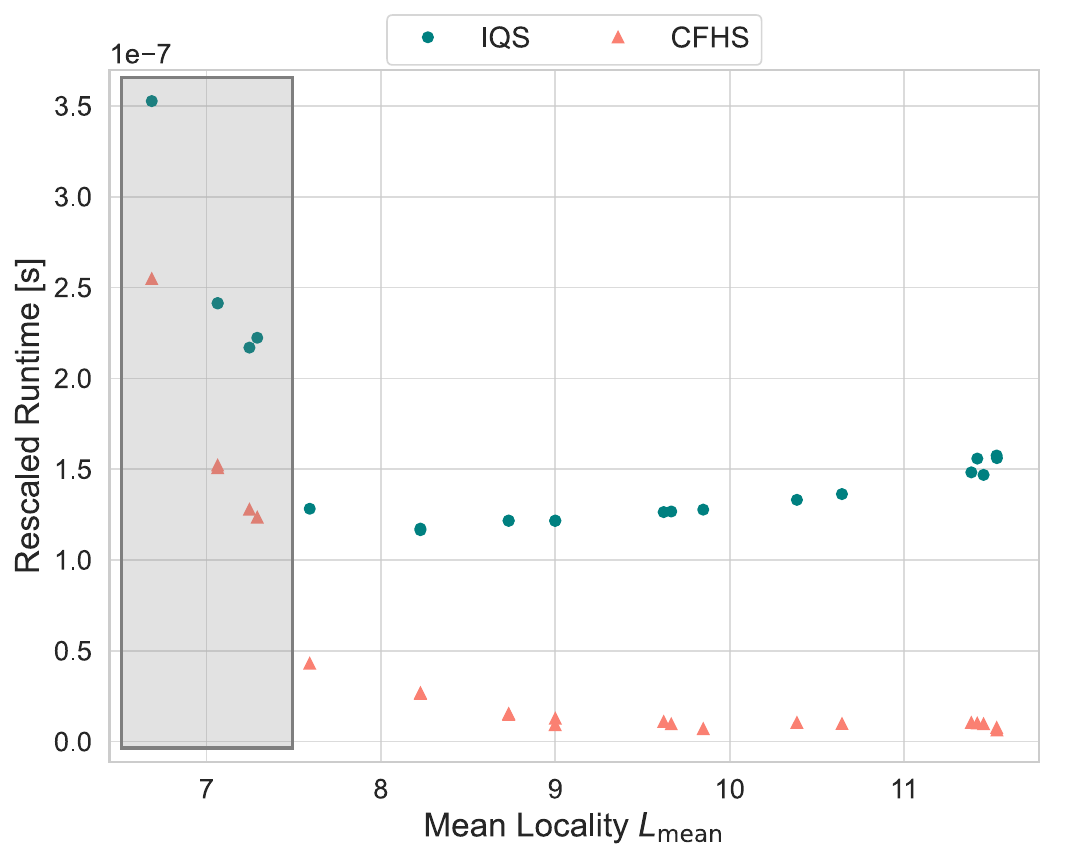}
        \caption{Rescaled runtime as in Eq.~\eqref{eqn:workload} without MPI usage.}\label{fig:workload_hamliba}
    \end{subfigure}%
    ~
    \begin{subfigure}[t]{0.49\textwidth}
        \centering
        \includegraphics[width=\textwidth]{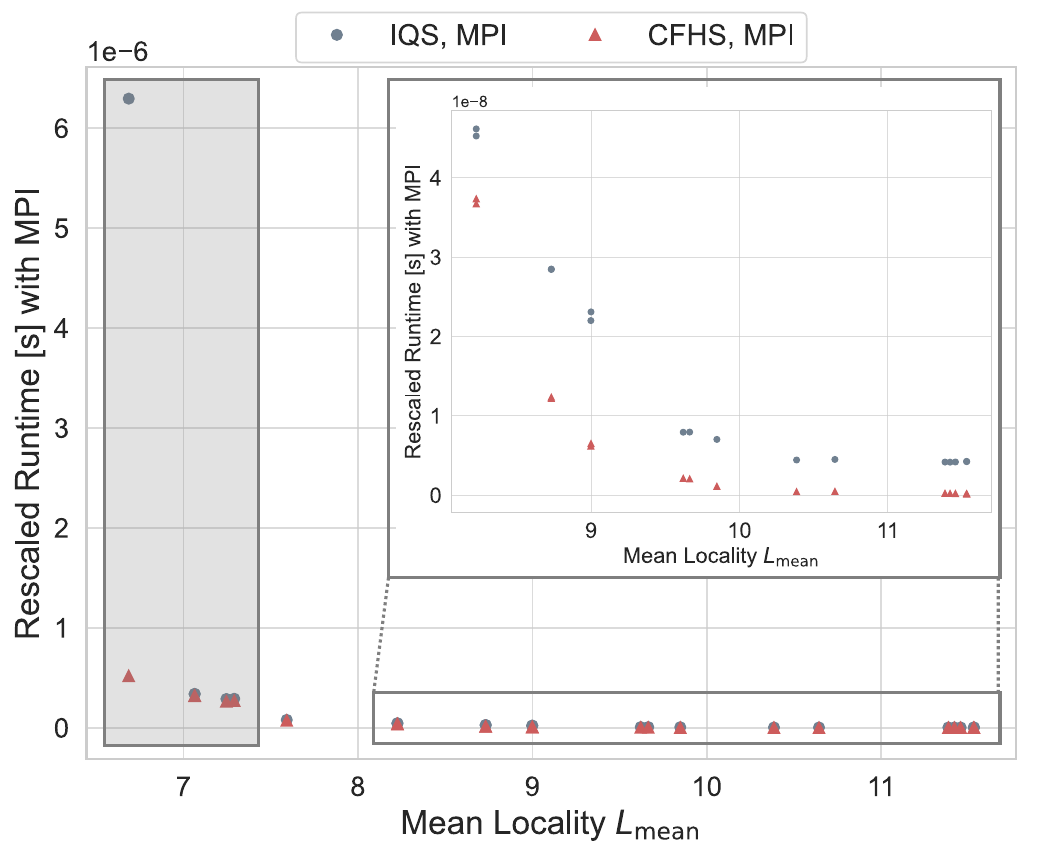}
        \caption{Rescaled runtime as in Eq.~\eqref{eqn:workload} with MPI parallelization. For better visibility, we present the results with~$L > 8$ in detail.}\label{fig:workload_hamlibb}
    \end{subfigure}
    \caption{Rescaled runtime [s] as defined in Eq.~\eqref{eqn:workload} of CFHS and IQS in dependence of the mean locality per Hamiltonian with and without MPI parallelization. Note the reduction in rescaled runtime achieved by CFHS compared to IQS in dependence on the mean locality $L_{\mathrm{mean}}$. As for the random Hamiltonians, we mark the regime in grey where the asymptotic behavior of the simulator is not yet achieved for small qubit numbers $n_{\text{qubits}}$, which again is more apparent for low locality $L$ since this is the only place where the small qubit numbers $n_{\text{qubits}}$ can contribute.}\label{fig:workload_hamlib}
\end{figure*}

\begin{landscape}
\subsection{Hamlib parameters and detailed evaluation results}\label{sec:appendix_eval}
\begin{table}[h!]
  \centering
  \caption{Numerical evaluation results and parameter configurations of various chemistry Hamiltonians. All Hamiltonians are part of HamLib~\cite{sawaya2024hamlib}. Exact Hamiltonian identifiers: $ES\_H12\_pyr =$ \texttt{ES\_H12\_pyramid\_R1.4\_sto\-6g\_ham}, $ES\_H12\_lin =$ \texttt{ES\_H12\_linear\_R2.0\_sto\-6g\_ham}, $ES\_H14\_ring =$ \texttt{ES\_H14\_ring\_R1.8\_sto\-6g\_ham}. The Hamiltonians are chosen from HamLib in order to cover a broad range with respect to the presented parameters: number of qubits ($n_{qubits}$), number of terms ($n_{\text{terms}}$). We present the compilation $t_{comp}$ and simulation time $t_{sim}$ results for IQS and CFHS. The uppercase abbreviation MPI denotes whether the Message-Passing Interface (MPI) was used for parallelization or not.}
  \label{tab:scores}
  \input{table.tex}

\end{table}
\end{landscape}

%% file: figures/CFHS_gate_flow.tex
\begin{tikzpicture}

\node[draw,
    rounded rectangle,
    minimum width=2.5cm,
    minimum height=1cm,
    fill=teal!20] (compute) {compute $G\left(F, \Phi \right)$};

\node[draw,
    rounded rectangle,
    minimum width=2.8cm,
    minimum height=2cm,
    below=of compute,
    fill=etonblue!20] (isclif) {Is gate $G$ Clifford?};

\node[draw,
    rectangle,
    below left=of isclif,
    fill=black!3] (applyclif) {$F \gets G(F)$};

\node[draw,
    rectangle,
    right =of isclif,
    fill=black!3] (getaxis) { $\sigma \gets \text{axis}(G)$ %
    $i \gets \text{qubit}(G)$ };

\node[draw,
    rounded rectangle,
    minimum width=2.8cm,
    minimum height=2cm,
    below =of getaxis,
    fill=etonblue!20] (iswhat1q) {What is $G$?};

\node[draw,
    rectangle,
    below left =of iswhat1q,
    fill=black!3] (applyprep) {$\Phi \gets \prep_{(\eff \sigma_i, \eff \overline \sigma_i)}(\Phi)$};

\node[draw,
    rectangle,
    below  =of iswhat1q,
    fill=black!3] (getangle) { $\theta \gets \text{angle}(G)$};

\node[draw,
    rectangle,
    below =of getangle,
    fill=black!3] (applyrot) {$\Phi \gets R_{\eff \sigma_i}(\theta)(\Phi)$};

\node[draw,
    rectangle,
    below right =of iswhat1q,
    fill=black!3] (applymeas) {$\Phi \gets \meas_{\eff \sigma_i}(\Phi)$};

\node[draw,
    rounded rectangle,
    minimum width=2.5cm,
    minimum height=1cm,
    below = 35em of compute,
    fill=teal!20] (return) {Return $\left(F, \Phi\right)$};

\node[draw,
    coordinate,
    below =of applyrot,
    fill=black!3] (endfullstate){};

\node[draw,
    coordinate,
    below = 32 em of compute] (end){};

\draw[-latex] (compute) edge (isclif);

\draw[-latex] (isclif) -| (applyclif)
    node[pos=0.25,fill=white,inner sep=0]{Yes};

\draw[-latex] (isclif) edge node[pos=0.4,fill=white,inner sep=0]{No} (getaxis)
(getaxis) edge (iswhat1q);

\draw[-latex] (iswhat1q) -| (applyprep)
    node[pos=0.25,fill=white,inner sep=0]{prep};

\draw[-latex] (iswhat1q) edge node[pos=0.25,fill=white,inner sep=0]{rot} (getangle)
(getangle) edge (applyrot);

\draw[-latex] (iswhat1q) -| (applymeas)
    node[pos=0.25,fill=white,inner sep=0]{meas};

\draw (applyclif) |- (end);

\draw (applyprep) |- (endfullstate);
\draw (applyrot) -- (endfullstate);
\draw (applymeas) |- (endfullstate);
\draw (endfullstate) |- (end);
\draw[-latex] (end) -- (return);

\end{tikzpicture}

%% file: table.tex
\rowcolors{2}{lightgray!40}{white}
\renewcommand{\arraystretch}{1.2} %

\begin{tabular}{@{} lrrrrrrrrrrrrr @{}}
\rowcolor{lightgray}
Name & $n_{\text{qubits}}$ & $n_{\text{terms}}$ & IQS $t_{comp}$  & IQS $t_{sim}$ & IQS $t_{comp}^{\text{MPI}}$  & IQS $t_{sim}^{\text{MPI}}$ & CFHS $t_{comp}$  & CFHS $t_{sim}$  & CFHS$_{\text{MPI}}$ $t_{comp}^{\text{MPI}}$  & CFHS $t_{sim}^{\text{MPI}}$ &  $L_{mean}$ & $L_{std}$ & $L_{max}$ \\
Li2 & 14 & 670 & 4.25 & 3.87 & 4.6 & 69.11 & 4.6 & 2.8 & 2.79 & 5.7 & 6.69 & 3.21 & 14 \\
CH & 14 & 1086 & 3.86 & 4.29 & 3.89 & 6.03 & 3.71 & 2.68 & 3.8 & 5.66 & 7.06 & 3.08 & 14 \\
NH & 14 & 1086 & 3.85 & 4.3 & 3.82 & 6.05 & 3.66 & 2.71 & 3.88 & 5.72 & 7.06 & 3.08 & 14 \\
NaLi & 14 & 1270 & 4.34 & 4.63 & 4.4 & 6.08 & 4.18 & 2.57 & 4.32 & 5.63 & 7.29 & 3.4 & 14 \\
OH & 14 & 1290 & 4.43 & 4.58 & 4.48 & 6.12 & 4.29 & 2.7 & 4.49 & 5.59 & 7.25 & 3.4 & 14 \\
Be2 & 16 & 1177 & 4.31 & 9.89 & 4.37 & 6.18 & 4.3 & 3.31 & 4.39 & 5.6 & 7.59 & 3.72 & 16 \\
HF & 16 & 2329 & 8.18 & 17.9 & 8.29 & 7.04 & 7.87 & 4.08 & 8.33 & 5.7 & 8.23 & 3.2 & 16 \\
BH & 16 & 2329 & 8.14 & 17.77 & 8.27 & 6.91 & 7.93 & 4.01 & 8.38 & 5.6 & 8.23 & 3.2 & 16 \\
C2 & 18 & 1884 & 7.33 & 60.03 & 7.23 & 14.05 & 7.01 & 7.29 & 7.04 & 6.02 & 8.73 & 3.83 & 18 \\
N2 & 18 & 1884 & 7.13 & 60.1 & 7.09 & 14.07 & 6.86 & 7.54 & 7.17 & 6.1 & 8.73 & 4.25 & 18 \\
HF & 18 & 3772 & 13.65 & 120.33 & 13.92 & 22.84 & 13.23 & 12.58 & 13.64 & 6.41 & 9 & 3.51 & 18 \\
NaLi & 18 & 3720 & 13.51 & 118.48 & 13.49 & 21.45 & 13.07 & 8.92 & 13.56 & 6.02 & 9 & 3.58 & 18 \\
NaLi & 20 & 5799 & 22.66 & 776.52 & 22.34 & 42.71 & 21.51 & 42.08 & 22.25 & 6.8 & 9.85 & 3.96 & 20 \\
H2 & 20 & 2951 & 11.87 & 391 & 11.68 & 24.54 & 11 & 33.62 & 11.13 & 6.54 & 9.62 & 4.25 & 20 \\
Na2 & 20 & 2951 & 11.87 & 392.02 & 11.69 & 24.6 & 11.18 & 29.64 & 11.46 & 6.31 & 9.66 & 4.71 & 20 \\
F2 & 22 & 4578 & 19.09 & 2616.76 & 19.31 & 86.42 & 18.27 & 186.39 & 19.17 & 8.91 & 10.64 & 4.58 & 22 \\
O3 & 22 & 5466 & 22.26 & 3051.4 & 22.09 & 101.57 & 21.15 & 237.89 & 22.62 & 9.85 & 10.38 & 4.6 & 22 \\
NaLi & 24 & 12685 & 61.45 & 33215.06 & 58.61 & 900.93 & 59.47 & 1318.17 & 59.2 & 29.95 & 11.53 & 4.71 & 24 \\
Na2 & 24 & 6509 & 29.06 & 17201.2 & 28.76 & 463.83 & 27.8 & 828.43 & 29.52 & 20.48 & 11.53 & 4.91 & 24 \\
O3 & 24 & 7881 & 34.55 & 19598.68 & 34.73 & 551.62 & 33.47 & 1363.18 & 35.01 & 30.98 & 11.38 & 5.05 & 24 \\
ES\_H12\_\textit{pyr} & 24 & 12453 & 61.5 & 32556.15 & 57.35 & 868.54 & 58.09 & 2090.97 & 57.09 & 44.09 & 11.42 & 4.54 & 24 \\
ES\_H12\_\textit{lin} & 24 & 14905 & 71.9 & 36719.32 & 70.64 & 1045.36 & 66.87 & 2406.17 & 70.48 & 50.31 & 11.45 & 4.76 & 24 \\
\end{tabular}